%% file: main.tex
\documentclass[runningheads]{llncs}

\input{macros.tex}
\begin{document}

\setlength{\textfloatsep}{8pt} 
\setlength{\abovecaptionskip}{8pt}
\setlength{\belowcaptionskip}{8pt}
\title{ Divide \& Scale:\\
Formalization and Roadmap to Robust Sharding}
%Formalization of the  Distributed Ledger Sharding Protocols
\titlerunning{Divide \& Scale}
% If the paper title is too long for the running head, you can set
% an abbreviated paper title here
%
 \author{Zeta Avarikioti\inst{1} \and
Antoine Desjardins\inst{2} \and
 Lefteris Kokoris-Kogias\inst{2,4}\and
 Roger Wattenhofer\inst{3}}
\authorrunning{Z. Avarikioti et al.}
% First names are abbreviated in the running head.
% If there are more than two authors, 'et al.' is used.
%
\institute{TU Wien, Austria \and ISTA, Austria \and
ETH  Z{\"u}rich, Switzerland \and Mysten Labs}
% \email{}\\
% \url{http://www.springer.com/gp/computer-science/lncs})
%
\maketitle              % typeset the header of the contribution
\begin{abstract}
Sharding distributed ledgers is a promising on-chain solution for scaling blockchains but lacks formal grounds, nurturing skepticism on whether such complex systems can scale blockchains securely. 
We fill this gap by introducing the first formal framework as well as a roadmap to robust sharding.
In particular, we first define the properties sharded distributed ledgers should fulfill. We build upon and extend the Bitcoin backbone protocol by defining \textit{consistency} and \textit{scalability}.
Consistency encompasses the need for atomic execution of cross-shard transactions to preserve safety, whereas scalability encapsulates the speedup a sharded system can gain in comparison to a non-sharded system.

Using our model, we explore the limitations of sharding. We show that a sharded ledger with $n$ participants cannot scale under a fully adaptive adversary, but it can scale up to $m$ shards where $n=c'm\log m$,
% $\frac{n}{c'\log{\frac{n}{c'\log n}}}$ 
under an epoch-adaptive adversary; the constant $c'$ encompasses the trade-off between security and scalability.  
This is possible only if the sharded ledgers create succinct proofs of the valid state updates at every epoch.
We leverage our results to identify the sufficient components for robust sharding, which we incorporate in a protocol abstraction termed Divide \& Scale.
To demonstrate the power of our framework, we analyze the most prominent sharded blockchains (\el, \mx, \ol, \rc) and pinpoint where they fail to meet the desired properties. 

% Sharding distributed ledgers is a promising on-chain solution for scaling blockchains. In this work, we define and analyze the properties a sharded distributed ledger should fulfill. More specifically, we show that a sharded blockchain with $n$ participants cannot be scalable under a fully adaptive adversary, but it can scale up to $\frac{n}{c\log{\frac{n}{c\log n}}}$under an epoch-adaptive adversary. This is possible only if the distributed ledger creates succinct proofs of the valid state updates at the end of each epoch. Our model builds upon and extends the Bitcoin backbone protocol by defining \textit{consistency} and \textit{scalability}. Consistency encompasses the need for atomic execution of cross-shard transactions to preserve safety, whereas scalability encapsulates the speedup a sharded system can gain in comparison to a non-sharded system. We introduce a protocol abstraction and highlight the sufficient components for secure and efficient sharding in our model. In order to show the power of our framework, we analyze the most prominent sharded blockchains (\el, \mx, \ol, \rc) and pinpoint where they fail to meet the desired properties. 

\keywords{Blockchains   \and Sharding \and Scalability \and Formalization.}
\end{abstract}
%
%
%
\input{1intro.tex}

\input{2model.tex}
\input{3limit-short.tex}

\input{4roadmap.tex}
\input{5comp-short.tex}

\newpage
\input{9figures.tex}
\input{ack}
% \newpage
%
% ---- Bibliography ----
%
% BibTeX users should specify bibliography style 'splncs04'.
% References will then be sorted and formatted in the correct style.
%

\input{references.bbl}
%

\newpage
\appendix
%

\input{3limit.tex}
\input{4analysis.tex}
\input{8eval-new.tex}

\end{document}

%% file: macros.tex
\usepackage{graphicx,color,comment,amsmath,soul,tikz,setspace,float,xspace}
\usepackage{hyperref}
\usepackage{amsfonts}
 \usepackage[]{todonotes}
\usepackage{enumitem} 
\usepackage{thm-restate}
\usepackage{pifont,bm}
\usepackage{caption}
\usepackage{subcaption}
\usepackage{multirow,footnote}
\usepackage{siunitx,booktabs}
\makesavenoteenv{table}
\usepackage[algoruled,linesnumbered]{algorithm2e}

\newcommand{\cmark}{\ding{51}}%
\newcommand{\xmark}{\ding{55}}%

\spnewtheorem{thrm}{Theorem}{\bfseries}{\itshape}
\spnewtheorem{defn}[thrm]{Definition}{\bfseries}{\rmfamily}
\spnewtheorem{lem}[thrm]{Lemma}{\bfseries}{\itshape}
\spnewtheorem{cor}[thrm]{Corollary}{\bfseries}{\itshape}
\spnewtheorem{assumption}[thrm]{Assumption}{\bfseries}{\itshape}
\newcommand{\com}[1]{}

\newcommand{\ol}{OmniLedger\xspace}
\newcommand{\mx}{Monoxide\xspace}
\newcommand{\rc}{RapidChain\xspace}
\newcommand{\el}{Elastico\xspace}

\newcommand{\cs}{Chainspace\xspace}

\newcommand{\ie}{i.\,e.\xspace}

\newcommand{\eg}{e.\,g.\xspace}

\usepackage{todonotes}
\usepackage{cleveref}

%% file: 1intro.tex
%------------------------------------------------------------------------
%-------------------------INTRO--------------------------------------
%----------------------------------------------------------------------
% \vspace{-0.5cm}
\section{Introduction}\label{sec:intro}
% \subsection{Motivation}
A promising solution to scaling blockchain protocols is sharding, \eg~\cite{zamani2018rapidchain,luu2016secure,wang2019monoxide,kokoris2017omniledger}. Its high-level idea is to employ multiple blockchains in parallel, the \textit{shards}, that operate using the same consensus protocol. Different sets of participants run consensus and validate transactions, so that the system ``scales out''.
% However, there is no formal definition for a robust sharded ledger (similar to the definition of what a robust transaction ledger is~\cite{garay2015bitcoin}). This leads to each new protocol defining its own set of goals, which tend to favor the presented system design. This in turn has lead to sharding often being criticized as
% some believe that the overhead of transactions between shards cancels out the potential benefits.

However, there is no formal definition of a robust sharded ledger (similar to the definition of what a robust transaction ledger is~\cite{garay2015bitcoin}), which leads to multiple problems.
First,  each protocol defines its own set of goals, which tend to favor the protocol design presented.  These goals are then confirmed achievable by experimental evaluations that demonstrate their improvements. 
Additionally, due to the lack of robust comparisons (which cannot cover all possible Byzantine behaviors), sharding is often criticized as some believe that the overhead of transactions between shards cancels out the potential benefits.
In order to fundamentally understand sharding, one must formally define what sharding really is, and then see whether different sharding techniques live up to their promise.

\vspace{-11pt}
\subsubsection*{Related work.}
Recently, a few systemizations of knowledge on sharding~\cite{wang2019sharding}, consensus~\cite{bano2017consensus}, and cross-shard communication~\cite{zamyatin2019sok} which have also discussed part of sharding, have emerged. These works, however, do not define sharding in a formal fashion to enable an ``apples-to-apples'' comparison of existing works  nor do they explore its limitations.

There are very few works that lay formal foundations for blockchain protocols. In particular, the Bitcoin backbone protocol~\cite{garay2015bitcoin} was the first to formally define and prove a blockchain protocol, specifically Bitcoin, in a PoW setting. Later, Pass et al.~\cite{pass2017analysis} showed that there is no PoW protocol that can be robust under asynchrony. 
With Ouroboros~\cite{kiayias2017ouroboros} Kiayias et al.\ extended the ideas of backbone to the Proof-of-Stake (PoS) setting, where they showed that it is possible to have a robust transaction ledger in a semi-synchronous environment as well~\cite{david2018ouroboros}. 
However, all of these works consider only non-sharded ledgers but can be used as stepping stones to the formalization of sharded ledgers.

\vspace{-11pt}
\subsubsection*{Our contribution.}
In this work, we take up the challenge of providing formal ``common grounds'' under which we can capture the sharding limitations, determine the necessary components of a sharding system, and fairly compare different sharding solutions.
We achieve this by defining a formal sharding framework as well as formal bounds of what a sharded transaction ledger can achieve.
% Then, we apply our framework to multiple sharded ledgers and identify why they do not satisfy our definition of a robust sharded transaction ledger. 
% %To achieve this we need to address two key challenges. First, we need to maintain compatibility with the existing models of a robust  transaction ledger introduced by Garay et al.~\cite{garay2015bitcoin}, so that the sharding framework constitutes a strict generalization. Second, our model needs to maintain enough relation with multiple existing sharded ledgers in order to be of practical use. 

To maintain compatibility with the existing models of a robust  transaction ledger, we build upon the work of Garay et al.~\cite{garay2015bitcoin}. We generalize the transaction ledger properties, originally introduced in \cite{garay2015bitcoin}, namely
\textit{Persistence} and \textit{Liveness}, to also apply to sharded ledgers. Persistence expresses the agreement between honest parties on the transaction order, while liveness encompasses that a transaction will eventually be processed and included in the transaction ledger.
Further, we extend the model to capture what sharding offers to blockchain systems by defining \textit{Consistency} and \textit{Scalability}. Consistency is a security property that conveys the atomic property of \textit{cross-shard transactions} (transactions that span multiple shards and should either abort or commit in all shards). 
% Scalability, on the other hand, is a performance property that encapsulates the speedup of a sharded blockchain system compared to a non-sharded blockchain system.
Scalability, on the other hand, is a performance property that encapsulates the resource gains per party 
(in bandwidth, storage, and computation) 
in a sharded system compared to a non-sharded system.
% Lastly, we define a theoretical performance metric, termed the \textit{throughput factor}, that expresses asymptotically the average number of transactions that can be processed per round in a sharding system under the worst possible Byzantine behavior.

Once we define the properties, we explore the limitations of sharding protocols that satisfy them.
% We observe a fine balance between consistency and scalability.
% To truly work, a sharding protocol should scale in computation, bandwidth, and storage.
We identify a  \emph{trade-off between the bandwidth requirements and how adaptive the adversary is}, \ie, how ``quickly'' the adversary can change the corrupted parties.
Specifically, with a fully adaptive adversary, scalable and secure sharding is impossible in our model. 
With a slowly-adaptive adversary, however, {sharding can scale securely} with up to $m$ shards, where $n=c' m\log m$. The constant $c'$ encompasses the trade-off between scalability and security: if the overall and per-shard adversarial thresholds are close to each other, then $c'$ must be large to  ensure security within each hard.  
Furthermore, scaling against a somewhat adaptive adversary is only possible under two  conditions: first, the parties of a shard \emph{cannot be light clients} to other shards to scale storage. Second,  shards must periodically \emph{compact the state} updates in a verifiable and succinct manner (\eg, via checkpoints~\cite{kokoris2017omniledger}, cryptographic accumulators~\cite{boneh2019batching}, zero-knowledge proofs~\cite{zkproofs2020,meckler2018coda} or other techniques~\cite{kiayias2020non,bunz2019flyclient,assimakis2019proof,kadhe2019sef}); else eventually the bandwidth resources per party will exceed those of a non-sharded blockchain.

Once we  provide solid bounds on the design of sharding protocols, we identify seven  components that are critical to designing a robust permissionless sharded ledger:
(a)  a core consensus protocol for each shard,
(b) a protocol to partition transactions in shards, 
(c) an atomic cross-shard communication protocol that enables transferring of value across shards,
(d) a Sybil-resistance mechanism that forces the adversary to commit resources in order to participate, 
(e) a process that guarantees honest and adversarial nodes  are appropriately dispersed to the shards to defend security against adversarial adaptivity,
(f) a distributed randomness generation protocol, 
(g) a process to occasionally compact the state in a verifiable manner.
We then employ these components to introduce a protocol abstraction, termed \textit{Divide \& Scale}, that achieves robust sharding in our model. We explain the design rationale, provide security proofs, and identify which components affect the  scalability and throughput of our protocol abstraction.

To demonstrate the power of our framework, we further describe, abstract, and analyze the most well-established
% \footnote{We consider all sharding protocols presented in top tier computer science conferences.} 
permisionless sharding protocols:
% Specifically, we describe, abstract, and analyze
\el~\cite{luu2016secure} (inspiration of \href{https://zilliqa.com/}{Zilliqa}),  \ol~\cite{kokoris2017omniledger} (inspiration of \href{https://harmony.one/}{Harmony}), \mx~\cite{wang2019monoxide}, and \rc~\cite{zamani2018rapidchain}. 
We demonstrate that all sharding systems fail to meet the desired properties in our model. \el and \mx do not actually (asymptotically) improve on storage over non-sharded blockchains according to our model. \ol is susceptible to a liveness attack where the adaptive adversary can simply delete a shard's state effectively preventing the system's progress. Albeit, with a simple fix, \ol satisfies all the desired properties in our model. Last, we prove \rc meets the desired properties but only in a weaker adversarial model.
For all protocols, we provide elaborate proofs while for \ol and \rc we further estimate how much they improve over their blockchain substrate. 
To that end, we define and use a theoretical performance metric, termed \textit{throughput factor}, which expresses the average number of transactions that can be processed per round under the worst possible Byzantine behavior.
We show that both \ol and \rc scale optimally with $m$ shards where $n=O(m \log m)$.

\paragraph{}In summary, the contribution of this work is the following: 
% \vspace{-4pt}
\begin{itemize}
    \item We introduce a framework where sharded transaction ledgers are formalized and the necessary properties of sharding protocols are defined. Further, we define a throughput factor to estimate the transaction throughput improvement of sharding blockchains over non-sharded blockchains (Section \ref{sec:model}).
    \item We explore the limitations of secure and efficient sharding protocols under our model (Section \ref{sec:limit}, Appendix~\ref{app:limit}).
    \item We identify the critical and sufficient ingredients  for designing a robust sharded ledger, which we incorporate into a protocol abstraction for robust sharding, termed Divide \& Scale (Section \ref{sec:roadmap}, Appendix~\ref{app:analysis}).
    % We further present  Divide \& Scale, a protocol abstraction for secure sharding, demonstrating how these ingredients yield a robust sharded ledger.
    % \item We evaluate existing protocols in our model. 
    % We show that \el, \mx, and \ol fail to satisfy some of the defined properties in our model, whereas \rc satisfies all the necessary properties to maintain a robust sharded transaction ledger but only under a weaker adversarial model (Section~\ref{app:comaprison}, Appendix~\ref{sec:eval}).
    \item We evaluate \el, \mx, \ol, and \rc. We pinpoint where the former three fail to satisfy our properties, whereas the latter satisfies them all only under a weaker adversarial model  (Section~\ref{app:comaprison}, Appendix~\ref{sec:eval}).
\end{itemize}
% Omitted proofs can be found in the Appendices.

% \subsection{Paper overview}

%% file: 2model.tex
%------------------------------------------------------------------------
%-------------------------MODEL--------------------------------------
%----------------------------------------------------------------------
\section{The sharding framework}\label{sec:model}
In this section,
% we introduce a formal definition of sharded transaction ledgers. In particular, 
we define the desired security and performance properties
of a secure and efficient distributed sharded ledger, extending the work of  Garay et al.~\cite{garay2015bitcoin}.
% Given these properties, we later explore the limitations of sharding, introduce a protocol abstraction that satisfies them, and evaluate if the existing sharding protocols maintain secure and efficient sharded transaction ledgers. 
We further define a theoretical performance metric, the transaction throughput.
To assist the reader, we provide a glossary of the most frequently used parameters in Table~\ref{tab:glossary} (Section Figures \ref{sec:figures}).

\subsection{The Model}

\paragraph{Network model.}
We analyze blockchain protocols assuming a synchronous communication network. In particular, a protocol proceeds in \textit{rounds}, and at the end of each round the participants of the protocol are able to synchronize, and all messages are delivered. 
A set of $R$ consecutive rounds $E=\{r_1,r_2,\dots,r_R \}$ defines an \textit{epoch}.
We consider a fixed number of participants in the system denoted by $n$. However, this number might not be known to the parties.
\vspace{-5pt}
\paragraph{Threat model.}
The adversary is slowly-adaptive, meaning that the adversary can corrupt parties on the fly at the beginning of each epoch but cannot change the corrupted set during the epoch, \ie, the adversary is static during each epoch. In addition, in any round, the adversary decides its strategy after receiving all honest parties' messages. The adversary can change the order of the honest parties' messages but cannot modify or drop them.
Furthermore, the adversary is computationally bounded and can corrupt at most $f$ parties during each epoch. This bound $f$ holds strictly at every round of the protocol execution.
Note that depending on the specifications of each protocol, \ie, which Sybil-attack-resistant mechanism is employed, the value $f$ represents a different manifestation of the adversary's power (\eg, computational power, stake in the system).
\vspace{-5pt}
\paragraph{Transaction model.}\label{subsubsec:txmodel}
We assume transactions consist of inputs and outputs that can only be spent as a whole. Each transaction input is an unspent transaction output (UTXO). Thus, a transaction takes UTXOs as inputs, destroys them and creates new UTXOs, the outputs. A transaction ledger that handles such transactions is UTXO-based, similarly to Bitcoin \cite{nakamoto2008bitcoin}. Most  protocols considered in this work are  UTXO-based. 
Transactions can have multiple inputs and outputs. 
We define the \textit{average size of a transaction}, i.e., the average number of inputs and outputs of a transaction in a transaction set, as a parameter $v$. This way $v$ correlates to the number of shards a transaction is expected to affect; the actual size in bytes is proportional to $v$ but unimportant for measuring scalability.
% \label{fn:txsize}. 
Further, we assume a transaction set $T$ follows a distribution $D_T$ (e.g.\ $D_T$ is the uniform distribution if the sender(s) and receiver(s) of each transaction are chosen uniformly at random from all possible users).
% All sharding protocols considered in this work are  UTXO-based.

%-----------------------SHARDED TL------------------------------
\subsection{Sharded Transaction Ledgers}
\label{subsec:ledger}
% In this section, we define what a robust sharded transaction ledger is.
% We build upon the definition of a robust public transaction ledger introduced in~\cite{garay2015bitcoin}.
% Then, we introduce the necessary properties a sharding blockchain protocol must satisfy in order to maintain a robust sharded transaction ledger.
In this section, we introduce the necessary properties a sharding blockchain protocol must satisfy in order to maintain a robust sharded transaction ledger. We build upon the definition of a robust transaction ledger introduced in~\cite{garay2015bitcoin}.

A sharded transaction ledger is defined with respect to a set of valid\footnote{Validity depends on the application using the ledger.} transactions $T$ and a collection of  transaction ledgers for each shard $S=\{S_1, S_2, \dots, S_m \}$. 
In each shard $i\in [m]=\{1,2, \dots,m \}$, a transaction ledger is defined with respect to a set of valid ledgers\footnote{Only one of the ledgers is actually committed as part of the shard's ledger, but before commitment there are multiple potential ledgers.} $S_i$ and a set of valid transactions. Each set possesses an efficient membership test. 
A ledger $L \in S_i$ is a vector of sequences of transactions $L = \langle x_1, x_2, \dots, x_l\rangle $, where $tx \in x_j \Rightarrow tx\in T, \forall j\in[l]$. 

% Thereby, we define a sharded transaction ledger as the collection of the transaction ledgers of each shard. Formally, a sharded transaction ledger is defined as a set of valid transactions $T$ and a collection of sets $S=\{S_1, S_2, \dots, S_m \}$, where each set $S_i, i\in [m]=\{1,2, \dots,m \}$ represents a set of valid ledgers of a shard. 
% A sharded transaction ledger is defined as a collection of sets $S=\{S_1, S_2, \dots, S_m \}$, where each set $S_i, i\in [m]=\{1,2, \dots,m \}$ represents valid ledgers of a shard, and a set of transactions $T$, each one possessing an efficient membership test.

In a sharding blockchain protocol, a sequence of transactions $x_i=tx_1\dots tx_e$ is inserted in a block which is appended to a party's local chain $C$ in a shard.
A chain $C$ of length $l$ contains the ledger $L_C=\langle x_1, x_2, \dots, x_l\rangle$ if the input of the $j$-th block in $C$ is $x_j$. 
The \textit{position} of transaction $tx_j$ in the ledger of a shard $L_C$ is the pair $(i, j)$ where $x_i = tx_1\dots tx_j \dots tx_e$ (\ie, the block that contains the transaction). Essentially, a party \textit{reports} a transaction $tx_j$ in position $i$ only if one of their shards' local ledger  includes transaction $tx_j$ in the $i$-th block.
We assume that a block has constant size, \ie, there is a maximum constant number of transactions included in each block\footnote{To scale in bandwidth, the block size cannot depend on the parties or transactions.}.
% We assume that a block has infinite space, and thus a transaction will never be excluded due to space limitations.

Furthermore, we define a symmetric relation on $T$, denoted by $M(\cdot , \cdot)$, that indicates if two transactions are conflicting, \ie, $M(tx,tx')=1 \Leftrightarrow tx, tx'$ are conflicting. 
Note that valid ledgers can never contain conflicting transactions. Similarly, a valid sharded ledger cannot contain two conflicting transactions even across shards. 
In our model, we assume there exists a verification oracle denoted by $V(T, S)$, which instantly verifies the validity of a transaction with respect to a ledger. In essence, the oracle $V$ takes as input a transaction $tx \in T$ and a valid ledger $L= \langle x_1, x_2, \dots, x_l\rangle \in S$ and checks whether the transaction is valid and not conflicting in this ledger; formally, $V(tx, L)=1 \Leftrightarrow \exists tx' \in L $ s.t. $M(tx,tx')=1$ or $L' = \langle x_1, x_2, \dots, x_l, tx \rangle$ is an invalid ledger.

Next, we introduce the security and performance properties a blockchain protocol must uphold to maintain a robust and efficient sharded transaction ledger: 
persistence, consistency, liveness, and  scalability.
Intuitively, \textit{persistence} expresses the agreement between honest parties on the transaction order, whereas \textit{consistency} conveys that cross-shard transactions are either committed or aborted atomically (in all shards).
\textit{Liveness} indicates that transactions will eventually be included in a shard, \ie, the system makes progress. Last, \textit{scalability} encapsulates the speedup of a sharded system in comparison to a non-sharded system: The blockchain's throughput limitation stems from the need for data propagation, maintenance, and verification by every party. Thus, to scale via sharding, each party must broadcast, maintain and verify mainly local information. 

\begin{defn}[Persistence]
\label{def:persistence}
Parameterized by $k \in \mathbb{N}$ (``depth'' parameter), if in a certain round an honest party reports a shard that contains a transaction $tx$ in a block at least $k$ blocks away from the end of the shard's ledger (such transaction will be called ``stable''), then whenever $tx$ is reported by any honest party it will be in the same position in the shard's ledger.
\end{defn}

\begin{defn}[Consistency]
\label{def:atomic}
Parametrized by $k \in \mathbb{N}$ (``depth'' parameter), 
% $v \in \mathbb{N}$ (average size of transactions) and $D_T$ (distribution of the input set of transactions), 
there is no round $r$ in which there are two honest parties $P_1, P_2$ reporting transactions $tx_1, tx_2$ respectively as stable (at least in depth $k$ in the respective shards), such that $M(tx_1,tx_2)=1$.
\end{defn}
Both persistence and consistency are necessary properties because one may fail while the other holds. For instance, if a party double-spends across two shards without reverting a stable transaction (e.g., due to a badly designed mechanism to process cross-shard transactions), consistency fails while persistence holds.

We further note consistency depends on the average size of transactions $v \in \mathbb{N}$ as well as the distribution of the input set of transactions $D_T$.  
For example, if all transactions are intra-shard, consistency is trivially satisfied due to persistence. 

To evaluate the system's progress, we assume that the block size is sufficiently large, thus a transaction will never be excluded due to space limitations.
\begin{defn}[Liveness]
\label{def:liveness}
Parameterized by $u$ (``wait time'') and $k$ (``depth'' parameter), 
% $v \in \mathbb{N}$ (average size of transactions) and $D_T$ (distribution of the input set of transactions), 
provided that a valid transaction is given as input to all honest parties of a shard continuously for the creation of $u$ consecutive rounds, then all honest parties will report this transaction at least $k$ blocks from the end of the shard, \ie, all report it as stable.
\end{defn}

Scaling distributed ledgers depends on three vectors: \textit{communication}, \textit{space}, and \textit{computation}. 
In particular, to allow high transaction throughput, the
bandwidth and computation required per party should ideally be constant and independent of the number of parties while the storage requirements per party should decrease with the number of parties. Such a system can scale optimally because an increased transaction load, e.g.\ double, can be processed with the same storage resources if the parties increase proportionally, e.g.\ double, as well as the same communication and computation resources per node. 
To measure scalability, i.e., the resource requirements per node, we define three scaling factors, namely the {communication, space, and computation factor}.

% \vspace{-5pt}
% \paragraph{Communication factor $\omega_m$:}
We define the \textit{\bf{communication factor $\omega_m$}} as the communication complexity of the system (per transaction) scaled over the number of participants. In essence, $\omega_m$ represents the average amount of sent or received data  (bandwidth) required per party to include  a transaction in the ledger. $\omega_m$ expresses the worst communication complexity of all the subroutines of the system, incorporating
 the bandwidth requirements of the protocols both within an epoch (\ie, within and across shards communication), as well as during epoch transitions (amortized over the epoch's length). The latter becomes the bottleneck for scalability in the long run as rotating parties must bootstrap to new shards and download the ever-growing shard ledgers.

% We observe though that in practice many protocols maintain low communication complexity, \eg, by employing efficient diffusion mechanisms of data such as gossip protocols, but fail to scale in the other dimensions. A notable example is the Bitcoin protocol where the bottleneck is space and computation, while the communication overhead is minimal.
%
We next introduce the \textit{\bf{space  factor $\omega_s$}} that estimates how much data each party stores in the system. To do so, we count the amount of data stored in total by all the parties scaled over the number of parties and the transaction load.
When $\omega_s$ is constant, $\Theta(1)$, each node stores all transactions equivalently to a central database, \eg, Bitcoin. On the contrary, a perfectly scalable system allows parties to share the  transaction load equally, $\omega_s = c/n$, $c$ constant; as a result, if parties increase proportionally to the transaction load the space resources per party remain the same.
% the space complexity of a sharding protocol per party, \ie, how much data is stored in total by all the participants of the system scaled over the in comparison to a single database that only stores the data once. 
% The space overhead factor ranges from $O(1)$ to $O(n)$, where $n$ is the fixed number of participants in the protocol.
% For instance, the Bitcoin protocol has  factor $O(n)$ since every party needs to store all data, while a perfectly scalable blockchain system has a constant space overhead factor.

To define the space factor we introduce the notion of average-case analysis. Typically, sharding protocols scale well when the analysis is optimistic, that is, for transaction inputs that  contain neither cross-shard nor multi-input (multi-output) transactions. However, in practice transactions are both cross-shard and multi-input/output. For this reason, we define the space factor as a random variable dependent on an input set of transactions $T$ drawn uniformly at random from a distribution $D_T$.

We assume $T$ is given well in advance as input to all parties. To be specific, we assume every transaction $tx \in T$ is given at least for $u$ consecutive rounds to all parties of the system. Hence, from the liveness property, all transaction ledgers held by honest parties will report all transactions in $T$ as stable.
Further, we denote by $L^{\lceil k}$ the vector $L$ where the last $k$ positions are ``pruned'', while $|L^{\lceil k}|$ denotes the number of transactions contained in this ``pruned'' ledger. We note that a similar notation holds for a chain $C$ where the last $k$ positions map to the last $k$ blocks.
Each party $P_j$ maintains a collection of ledgers $SL_j=\{L_1, L_2, \dots, L_s \}, 1 \leq s \leq m$. 
We may now define the \textit{space factor} for a sharding protocol with input $T$ as the number of stable transactions included in every party's collection of transaction ledgers over the number of parties $n$ and the number of input  transactions $T$\footnote{Without loss of generality, we assume all transactions are valid and thus are eventually included in all honest parties' ledgers.},
$
\omega_s(T)= \sum_{\forall j \in [n]} \sum_{\forall L \in SL_j} |L^{\lceil k}| / (n|T|)$.

Lastly, we consider the verification process which can be computationally expensive. In our model, we focus on the average verification cost per transaction. We assume a constant  computational cost per verification, \ie, a party's running time of verifying if a transaction is invalid or conflicting with a ledger is considered constant because this process can always speed up using efficient data structures (\eg\ trees allow for logarithmic lookup time). Thus, the computational cost of a party is defined by the number of times the party executes the verification process. 
For this purpose, we employ a verification oracle $V$.
%For this purpose, we employ a (verification) oracle, denoted by $V(T, S)$ which takes as input a transaction $tx \in T$ and a valid ledger $x \in S$ and checks whether the transaction is valid and not conflicting in this ledger; formally, $V(tx, x)=1 \Leftrightarrow \exists tx' \in x $ s.t. $M(tx,tx')=1$ or $x' \equiv x \cup tx$ is an invalid ledger. 
Each party calls the oracle to verify transactions, pending or included in a block. We denote by $q_i$ the number of times party $P_i$ calls oracle $V$ in a protocol execution.
The \textit{\bf{computational factor $\omega_c$}} reflects the total number of times all parties call the verification oracle in a protocol execution scaled over the number of transactions $T$, 
{$\omega_c (T) =  \sum_{\forall i \in [n]} q_i / |T|$}.

An ideal sharding system only involves a constant number of parties to verify each transaction, $\omega_c = \Theta(1)$, while both a typical BFT-based protocol and Bitcoin demand all nodes to verify all transactions,  $\omega_c = \Theta(n)$.
Furthermore, the computational  factor is a random variable, hence the objective is to calculate the expected value of $\omega_c$, \ie, the probability-weighted average of all possible values, where the probability is taken over the input transactions $T$.

Intuitively, scaling means processing more transactions with similar (\ie, not proportionally increasing) resources per party. 
If parties share the transaction load, \eg, space scales $\omega_s=c/n$, increased transactions can be processed by increasing the number of parties. Subsequently, the communication and computational costs must not increase proportionally to the number of parties, \ie,  $\omega_c=o(n)$ and $\omega_m=o(n)$, else the system cannot truly scale the transaction load. 
We observe, however, that in practice  protocols may scale well in one dimension but fail in another. A notable example is the Bitcoin protocol which has minimal communication overhead but does not scale in space and computation. 
To ensure overall scaling capabilities, we define the scalability property of sharded ledgers below; we say that a sharded ledger satisfies scalability if and only if the system scales in all the aforementioned dimensions. 
 
\begin{defn}[Scalability]
\label{def:scalability}
Parameterized by $n$ (number of participants), $v \in \mathbb{N}$ (average size of transactions), $D_T$ (distribution of the input set of transactions),
the communication, space and computational factors of a sharding blockchain protocol are $\omega_m=o(n)$, $\omega_s=o(1)$, and  $\omega_c=o(n)$, respectively. 
\end{defn}

In order to adhere to standard security proofs from now on we say that the protocol $\Pi$ \textit{satisfies} property $Q$ in our model if $Q$ holds with overwhelming probability (in a security parameter).
Note that a probability $p$ is overwhelming if $1-p$ is negligible. A function $negl(k)$ is negligible if for every $c>0$, there exists an $N>0$ such that $negl(k)<1/k^c$ for all $k>\geq N$. 
Furthermore, we denote by $\mathbb{E}(\cdot)$ the expected value of a random variable.

% Further, we say that the protocol $\Pi$ \textit{satisfies on expectation} property $Q$ in our model, if $Q$ holds on average in the execution of protocol $\Pi$.

\begin{defn}[Robust Sharded Transaction Ledger]\label{def:sharding}
A protocol that satisfies the properties of persistence, consistency, liveness, and scalability maintains a robust sharded transaction ledger.
\end{defn}
% We say that a sharding protocol that satisfies the properties of persistence, consistency, liveness, and scalability maintains a \textit{robust sharded transaction ledger}.

%-----------------------MODEL FOR PROTOCOLS------------------------------
\subsection{(Sharding) Blockchain Protocols} \label{subsec:cryptoprop}
In this section, we adopt the definitions and properties of~\cite{garay2015bitcoin} for blockchain protocols, while we slightly change the notation to fit our model. 
In particular, we assume the parties of a shard of any sharding protocol maintain a chain (ledger) to achieve consensus. This means that every shard internally executes a blockchain (consensus) protocol that has three properties as defined by \cite{garay2015bitcoin}: chain growth, chain quality, and common prefix. 
Each consensus protocol satisfies these properties with different parameters. 

In this work, we will use the properties of the shards' consensus protocol to prove that a sharding protocol maintains a robust sharded transaction ledger. In addition, we will specifically use the shard growth and shard quality parameters to estimate the transaction throughput of a sharding protocol.
The following definitions follow closely Definitions $3, 4$ and $5$ of \cite{garay2015bitcoin}.

\begin{defn}[Shard Growth Property]
\label{def:growth}
Parametrized by $\tau \in \mathbb{R}$ and $s \in \mathbb{N}$, for any honest party $P$ with chain $C$, it holds that for any $s$ rounds there are at least $\tau \cdot s$ blocks added to chain $C$ of $P$.
\end{defn}

\begin{defn}[Shard Quality Property]
\label{def:quality}
Parametrized by $\mu \in \mathbb{R}$ and $l\in \mathbb{N}$,  for any honest party $P$ with chain $C$, it holds that for any $l$ consecutive blocks of $C$ the ratio of honest blocks in $C$ is at least $\mu$.
\end{defn}

\begin{defn}[Common Prefix Property]
\label{def:commonprefix}
Parametrized by $k\in \mathbb{N}$, for any pair of honest parties $P_1, P_2$ adopting chains $C_1,C_2$ (in the same shard) at rounds $r_1 \leq r_2$  respectively, it holds that $C_1^{\lceil k}\preceq C_2$, where $\preceq$ denotes the prefix relation.
\end{defn}

Next, we define the \textit{\bf degree of parallelism} (DoP) of a sharding protocol, denoted $m'$. 
To evaluate the DoP of a protocol with input $T$, we need to determine how many shards are affected by each transaction on average; essentially, estimate how many times we run consensus for each valid transaction until it is stable. This is determined by the mechanism that handles the cross-shard transactions. 
To that end, we define $m_{i,j}=1$ if the $j$-th transaction of set $T$ has either an input or an output that is assigned to the $i$-th shard; otherwise $m_{i,j}=0$. Then, the DoP of a protocol execution over a set of transactions $T$ is defined as follows:
$m' = \frac{T \cdot m}{\sum_{j=1}^T \sum_{i=1}^m m_{i,j}}$.
The DoP of a protocol execution depends on the distribution of transactions $D_T$, the average size of transactions $v$, and the number of shards $m$.
For instance, assuming a uniform distribution $D_T$, the expected DoP is $\mathbb{E}(m')=m/v$.

We can now define an efficiency metric, the \textit{transaction throughput} of a sharding protocol. Considering constant block size, we have:
\begin{defn}[Throughput]
\label{def:throughput}
The expected transaction throughput in $s$ rounds of a sharding protocol with $m$ shards  is $\mu \cdot \tau \cdot s \cdot m'$.
We define the throughput factor of a sharding protocol $\sigma=\mu \cdot \tau \cdot m'$.
\end{defn}
Intuitively, the throughput factor expresses the average number of blocks that can be processed per round by a sharding protocol.
Thus, the transaction throughput (per round) can be determined by the block size multiplied by the throughput factor. 
%make this a proof in the final version - about the block size being constant
The block size is considered constant; however, it cannot be arbitrarily large. The limit on the block size is determined by the bandwidth of the ``slowest'' party within each shard.
At the same time, the constant block size guarantees low latency. If the block size is very large or depends on the number of shards or the number of participants, \textit{bandwidth or latency} becomes the performance bottleneck. 
As our goal is to estimate the efficiency of the transactions' parallelism in a protocol, other factors like cross-shard communication latency are omitted.

%% file: 3limit-short.tex
%----------------------------------------------------------------------
%-------------------------LOWER BOUNDS----------------------------------
%----------------------------------------------------------------------
\section{Limitations of sharding protocols}\label{sec:limit}

In this section, we present a summary of our analysis  on the limitations of sharding protocols in our framework (cf.\ Appendix~\ref{app:limit}). 

First, we focus on the limitations that stem from the nature of the transaction workload. In particular, sharding protocols are affected by two characteristics of the input transaction set: the transaction size $v$ (number of inputs and outputs of each transaction), and more importantly the number of cross-shard transactions.

The average size of transactions is fairly small in practice, e.g., an average Bitcoin transaction has $2$ inputs and $3$ outputs with a small deviation~\cite{bitcoinvisuals}. We thus assume a fixed number of UTXOs participating in each transaction, meaning the transaction size $v$ is a small constant. 
Furthermore, as $v$ increases, more shards are affected by each transaction on expectation, hence the number of cross-shard transactions increases. To meaningfully lower bound the ratio of cross-shard transactions, we thus consider the minimum transaction size $v=2$. If a transaction has more UTXOs, its chance of being cross-shard only increases. 

% Regarding cross-shard transactions, the outcome depends on the distribution of the input transaction set $D_T$, as well as the process that partitions transactions into shards.
The number of cross-shard transactions depends on the distribution of the input transactions $D_T$, as well as the process that partitions transactions into shards.
First, we assume each ledger interacts (\ie, shares a cross-shard transaction) with $\gamma$ other ledgers on average, $\gamma$ being a function dependent on the number of shards $m$.
We examine protocols where parties maintain information on shards other than their own and derive an upper bound for the expected value of $\gamma$ such that scalability holds.
Leveraging that, we prove the following:
% there is no protocol that maintains a robust sharded ledger against an adaptive adversary.

\begin{restatable}{thrm}{adaptive}\label{thm:adaptive}
There is no protocol maintaining a robust sharded transaction ledger against an adaptive adversary in our model controlling $f \geq n/m$, where $m$ is the number of shards, and $n$ is the number of parties.
\end{restatable}

Next, we extend our results assuming, similarly to most sharding systems, that the UTXO space is partitioned uniformly at random into shards. 
In particular, we first show that a constant fraction of transactions is expected to be cross-shard. Using that we demonstrate there is no sharded ledger that satisfies scalability if parties store any information on ledgers (other than their own) involved in cross-shard transactions, \ie, are light clients on other shards~\cite{chatzigiannis2022sok}. 
We stress that our results hold for \textit{any distribution} where the expected number of cross-shard transactions is \textit{proportional} to the number of shards.

% Later, we assume shards are created with a uniformly random process, \ie, the UTXO space is partitioned uniformly at random into shards. 
% Under this assumption (which most sharding systems make), we show a constant fraction of the transactions are expected to be cross-shard, and there is no sharded ledger that satisfies scalability if participants have to maintain any information on ledgers others than their own. 
% Note that our results hold for \textit{any distribution} where the expected number of cross-shard transactions is \textit{proportional} to the number of shards.
\begin{restatable}{thrm}{light}\label{thm:scale-SPV}
There is no protocol that maintains a robust sharded transaction ledger in our model under uniform space partition  when parties are light nodes on the shards involved in cross-shard transactions.
\end{restatable}

We further identify a concrete trade-off between security and scalability, that stems from the way parties are partitioned into shards. In particular, when parties are randomly permuted among shards, which is a common practice in sharding, \eg,~\cite{kokoris2017omniledger,luu2016secure}, sharding scales almost linearly. The trade-off is now captured by the constant $c'$:  if the overall and per-shard adversarial thresholds are close to each other, then $c'$ must be large to  ensure security within each shard.

% We further extend our model to consider the case where parties are assigned to shards independently and uniformly at random. 
% This assumption is met by almost all known sharding systems to guarantee security against non-static adversaries.
% We show that any sharding protocol can scale at most by a factor of $n/\log n$ in our model. 

\begin{restatable}{thrm}{scale}
Any protocol that maintains a robust sharded transaction ledger in our model under uniformly random partition of the state and parties, can scale at most by a factor of $m$, where $n=c'm\log m$ and the constant $c'$ encompasses the trade-off between security and scalability.
\end{restatable}

Finally, we demonstrate the importance of periodical compaction of the valid state-updates in sharding protocols: we prove that any sharding protocol that satisfies scalability in our model, when the state is uniformly partitioned and the parties are periodically shuffled among shards, requires a state-compaction process such as checkpoints~\cite{kokoris2017omniledger}, cryptographic accumulators~\cite{boneh2019batching}, zero-knowledge proofs~\cite{zkproofs2020}, non-interactive proofs of proofs-of-work~\cite{kiayias2020non,bunz2019flyclient}, proof of necessary work~\cite{assimakis2019proof}, erasure codes~\cite{kadhe2019sef}, etc.
Intuitively, parties must be periodically shuffled among shards to maintain security against adaptivity. Subsequently, the parties must occasionally bootstrap to the new ever-increasing blockchains, leading to bandwidth or storage overheads that exceed those of a non-sharded blockchain in the long run.
We stress that this result holds even if the parties are not randomly shuffled among the shards, as long as a significant fraction of parties changes shards from epoch to epoch.

\begin{restatable}{thrm}{checkpoints}\label{thm:checkpoints}
Any protocol that maintains  a robust sharded transaction ledger in our model, under uniformly random partition of the state and parties, employs verifiable compaction of the state.
\end{restatable}

%% file: 4roadmap.tex
%------------------------------------------------------------------------
%-------------------------Roadmap--------------------------------------
%----------------------------------------------------------------------

\section{Divide \& Scale} \label{sec:roadmap}
In this section, we discuss our design rationale for robust sharding; using the bounds of Section \ref{sec:limit}, we deduce some sufficient components for robust sharding in our model. We leverage these components to introduce a \textit{protocol abstraction} for robust sharding, termed \textit{Divide \& Scale}, in Algorithm~\ref{alg:shard-crux}. We prove Divide \& Scale is secure in our model (assuming the components are secure) and evaluate its efficiency depending on the choices of the individual components in Appendix~\ref{app:analysis}.

\vspace{-11pt}
\subsubsection*{Sharding Components.}
% \vspace{-8pt}
% \noindent\textbf{Sharding Components.}
 We explain our design rationale and introduce the ingredients of a  protocol that maintains a robust sharded ledger.
\vspace{-4pt}
\begin{enumerate}[wide,label=(\alph*), labelindent=2pt, itemsep=3pt
]

    \item \textbf{Consensus protocol of shards or \texttt{Consensus}:} 
    A sharding protocol either runs consensus in every shard separately (multi-consensus) or provides a single total ordering for all the blocks generated in each shard (uni-consensus~\cite{al2019lazyledger,rana2020free2shard}). Since uni-consensus takes polynomial cost per block, such a protocol can only scale if the block size is also polynomial (e.g., includes  $\Omega(n)$ transactions~\cite{rana2020free2shard}). However, in such a case, the resources of each node generating an $\Omega(n)$-sized block must also grow with $n$, and therefore scalability  cannot be satisfied\footnote{Due to their inherent inability to asymptotically scale, we believe uni-consensus systems are categorized as performance optimizations of consensus, \eg, \cite{androulaki18hyperledger,stathakopoulou2019mir,avarikioti2020fnfbft,danezis2022narwhal,spiegelman2022bullshark}.}. 
    For this reason, in our protocol abstraction, we chose the multi-consensus approach.
    
    The consensus protocol run per shard must satisfy the properties of Garay et al.~\cite{garay2015bitcoin}: \textit{common prefix, chain quality, and chain growth}. These properties are necessary (but not sufficient) to ensure persistence, liveness, and consistency.
    % with the additional requirement that light-client verification can be done with a sublinear proof (to the number of rounds), regardless of the last time the light-client synchronized with the chain.
    
    \item \textbf{Cross-shard mechanism or \texttt{CrossShard}:} 
    The cross-shard mechanism is the protocol that handles the transactions that span across multiple shards.
    It is critical for the security of the sharding system, as it guarantees consistency, as well as scalability; a naively designed cross-shard mechanism may induce high storage or communication overhead on the nodes when handling several cross-shard transactions. To that end, the limitations of Section \ref{sec:limit} apply.
 % to designing robust sharded  ledgers.
    
    The cross-shard mechanism should provide the \textit{ACID properties} (as in database transactions). Durability and Isolation are provided directly by the blockchains of the shards, hence, the cross-shard mechanism should provide {Consistency}, \ie, every transaction that commits produces a semantically valid state, and {Atomicity}, \ie, transactions are committed and aborted atomically (all or nothing).
    Typically the cross-shard mechanism runs hand in hand with the consensus protocol to guarantee consistency across shards.
    
    \item \textbf{Sybil-resistance mechanism or \texttt{Sybil}:} 
    The Sybil-resistance mechanism enables the participants of a permissionless setting to reach a global consensus on a set of fairly-selected valid identities. Its fair selection, \ie, assigning valid identities to each party proportionally to its spent resources, guarantees the security bounds of the consensus protocol (\eg, $f<1/3$ for BFT).
    To ensure fairness against slowly-adaptive adversaries, the Sybil-resistance mechanism must have access to unknown unbiasable randomness (see below DRG).
    The exact protocol (\eg\ PoW, PoS) is irrelevant to our analysis as long as it guarantees 
    (i)~\textit{correctness}: all parties can verify a valid identity, 
    (ii)~\textit{fairness}:  each party is selected with probability proportional to its resources,
    and (iii)~\textit{unpredictability}: no party can predict beforehand the valid set of identities (for the new epoch).
    % (i)~all parties agree on the same set of identities, (ii)~no malicious party can convince an honest party that its identity is valid for an epoch if it is not, (iii)~no private information is leaked (\eg, a party's secret key), and (iv)~for each party, the probability of being granted an identity is proportional to its invested resources. %and (v)~ the protocol terminates.}

    \item \texttt{\textbf{StatePartition:}} This protocol determines how the state (\eg\ transactions) is partitioned into shards. A naive design may violate consistency but there are several secure solutions to employ, \eg~\cite{kokoris2017omniledger,zamani2018rapidchain}. 
    We perform our analysis assuming all transactions are cross-shard, because any secure protocol that performs well in the pessimistic case, also performs well  when transactions are intra-shard. Moreover, in the latter case, scaling is not challenging as the transaction throughput can be processed securely in blockchains that work in parallel.  

    \item \textbf{Division of nodes to shards or \texttt{Divide2Shards}:} 
    This is the protocol that determines how parties are assigned to shards.
    It is crucial for security against slowly adaptive adversaries as a fully corrupted shard may result in the loss of all three security properties.
    It is also the reason that sharding cannot tolerate fully-adaptive adversaries in our model (Theorem~\ref{thm:adaptive}). Note that static adversaries are an easier subcase of the slowly adaptive one.

    In particular, to ensure transaction finality (\ie, liveness and persistence), either the consensus security bounds must hold for each shard, or the protocol must guarantee that if the adversary compromises a shard then the security violation will be restored within a specific (small) number of rounds. Specifically, if an adversary completely or partially compromises a shard, effectively violating the consensus bounds, then the adversary can double spend within the shard (violates persistence), as well as across shards (because nodes cannot verify cross-shard transactions from Lemma~\ref{lem:cross-bound}).
    Therefore, the transactions included in these blocks can only be executed when honest parties have verified them. Partial solutions towards this direction have been proposed such as proofs of fraud that allow an honest party to later prove misbehavior. Another challenge of this approach is to guarantee data availability.  
    
    Due to the complexity of such solutions and their implications on the transactions' finality, we design Divide \& Scale assuming  \textit{the security bounds of consensus are maintained} when parties are divided into shards.
    Specifically,  
    the parties are shuffled at the beginning of each epoch so that the threat model holds.
    % that at the beginning of each epoch, the nodes are shuffled among the shards to satisfy the consensus bounds. 
    A secure shuffling process requires an ubiasable source of randomness (see below DRG). 
    When assigned to a shard, the nodes update their local state with the state of the new shard they are asked to secure, which in turn affects scalability. 
    \textit{The frequency of shuffling is thereby incorporating the trade-off between scalability and adaptive security.}

    % In short, the \texttt{Divide2Shards} should guarantee every party knows its shard assignment, and the distribution of parties to shards  maintains the power of the adversary in each shard to some desired proportion of the initial bound (\eg, the adversary controls 25\% of the system and after the shuffling controls up to 33\% of each shard).
    
    % To summarize, the \texttt{Divide2Shards} protocol should guarantee that
    % (i)~every party knows its shard assignment, 
    % % (ii)~the protocol terminates efficiently with respect to the epoch size (in rounds), and 
    % (iii)~ given a bounded adversary, the distribution of parties to shards maintains the power of the adversary in each shard to some desired proportion of the initial bound (\eg, the adversary controls 25\% of the system and after the shuffling controls up to 33\% of each shard).

    \item \textbf{Randomness generation protocol or \texttt{DRG}:} 
    The DRG protocol \textit{provides unpredictable unbiasable randomness}~\cite{feldman1987practical,cascudo2017scrape,schindler2018hydrand,syta2017scalable,kokoris2020asynchronous,bonneau2015bitcoin,bunz2017proofs,das2022practical} such that both \texttt{Sybil} and \texttt{Divide2Shards} result in shards that maintain the security bounds for the consensus protocol.
    Given a slowly-adaptive adversary, the DRG protocol must be executed (at least) once per epoch; its high communication complexity can be amortized over the rounds of an epoch such that the system scales.
    % The protocol is executed either by a reference committee~\cite{zamani2018rapidchain} or by all the parties~\cite{kokoris2017omniledger} and its cost can be amortized over the rounds of an epoch.
    
    \item \textbf{Verifiable compaction of state or \texttt{CompactState}:}
    \texttt{CompactState} guarantees that periodically state updates can be  verifiably compacted. This protocol is necessary for scaling sharding systems in the long run, as it ensures that new parties can bootstrap with minimal effort (Theorem \ref{thm:checkpoints}).   
    % At the end of each epoch, the history of the shards' ledgers should be summarized such that new parties can bootstrap with minimal effort (Theorem \ref{thm:checkpoints}). To this end, a process for verifiable compaction of the state-updates must be employed.
    The compacted state must be broadcasted to all parties, \eg\ via reliable broadcast~\cite{bracha1985asynchronous}, to ensure data integrity and data availability; else a slowly-adaptive adversary can corrupt an entire shard after an epoch transition, violating liveness.
    % to the node that will be assigned to each shard in the next epoch.
    Any protocol that ensures data binding and data availability can be used.
    In summary, this protocol must guarantee (i)~\textit{verifiable asymptotic compression} (more than constant), and (ii)~\textit{data integrity and availability}, \ie, the ledgers' history is available and can be retrieved.
    % (i)~significant \textit{compression} (asymptotically more than constant), (ii)~verifiable, (iii)~the data can be retrieved (integrity), and (iv)~the data are available.} 
    To satisfy scalability, the protocol must also ensure (iii)~\textit{efficient communication complexity} with respect to the epoch size (in rounds).
    % This step must take place before the new epoch begins as the slowly adaptive adversary is allowed to compromise the nodes of a shard in the previous epoch, violating the liveness of the system.

\end{enumerate}
% Last, we assume that all protocols satisfy liveness in our security model.
% \vspace{-14pt}
% \subsubsection{Divide \& Scale.}
% We introduce a protocol abstraction that achieves robust sharding in our model in Algorithm~\ref{alg:shard-crux}. In the protocol design, we employ the  components of sharding as black boxes, and later prove security and analyze the overall performance depending on the choice of the components (Appendix~\ref{app:analysis}).

\begin{algorithm}[t!] 
\SetAlgorithmName{Protocol Abstraction}
\SetStartEndCondition{ }{}{}
\SetKwProg{Fn}{def}{\string:}{}
\SetKw{KwTo}{to}\SetKwFor{For}{for}{\string:}{}
\SetKwIF{If}{ElseIf}{Else}{if}{:}{elif}{else:}{}
\AlgoDontDisplayBlockMarkers\SetAlgoNoEnd\SetAlgoNoLine
\let\oldnl\nl% Store \nl in \oldnl
\newcommand{\nonl}{\renewcommand{\nl}{\let\nl\oldnl}}% Remove line number for one line
\small{
 \KwData{$N_0$ nodes are participating in the system at round $0$ (genesis block). $m(N_E)$ denotes the function that determines the number of shards in epoch $E$. The transactions of epoch $E$ are $T_E$. $i$ denotes the block round (its relation to the communication rounds depends on the employed components).}
 \nonl
 \KwResult{Shard state $T=\{T_0,T_1,\dots \}$.}\
 
 \tcc{Initialization}
 $i \leftarrow 1$\\
 $E\leftarrow 0$\\
 
  \tcc{Beginning of epoch: retrieve identities from Sybil resistant protocol, execute the DRG protocol to create the new epoch randomness, and assign nodes to shards}
  \uIf{$i\mod R = 1$}{ 
    $E \leftarrow E+1$\\
    \uIf{$i \neq 1$}{ $N_{E} \leftarrow$ \texttt{Sybil($r_{E-1}$)}
    }
    $r_{E} \leftarrow$ \texttt{DRG($N_E$)}\\
    Call \texttt{Divide2Shards($N_E, m(N_E), r_E$)}
  }
  
    \tcc{End of epoch: compact the state of the shard}
    \uElseIf{$i\mod R= 0$}{
        Call \texttt{CompactState(i)}
    }
    
    \tcc{During epoch: run the consensus protocol for intra-shard and cross-shard transactions}
    \Else{
        \uIf{If transaction $t \in T_E$ is cross-shard}{
            Call \texttt{CrossShard(t)} \tcp*{Invokes \texttt{Consensus} in multiple shards}
        }
         \Else{
            Call \texttt{Consensus(t)}
        }
    }
    
$i \leftarrow i+1$\\
Go to step 3
}
 \caption{Divide \& Scale}
 \label{alg:shard-crux}
\end{algorithm}

\vspace{-8pt}

%% file: 5comp-short.tex
\section{Evaluation of sharding protocols}\label{app:comaprison} 
% To showcase the wide applicability and value of our framework, we evaluate well-established sharding protocols,  in particular, \el, \mx, \ol, and \rc, with respect to the desired properties of Section \ref{subsec:ledger}. We further discuss \cs, which is permissioned.
To showcase the wide applicability and value of our framework, we evaluate in our model the well-established sharding protocols \el, \mx, \ol, and \rc, and discuss \cs.
We refer the reader to Appendix~~\ref{sec:eval} for the complete analysis where we identify each protocol's sharding components as defined in Section~\ref{sec:roadmap} which we use to prove or disprove the desired properties of Section~\ref{sec:model}, often leveraging the bounds of Section~\ref{sec:limit}.
Due to space limitations, we only discuss here the final results of our analysis, also illustrated in  Table~\ref{tab:comparison},  with key insights on how each protocol fails to meet some of the properties.
%
 % We evaluate well-established sharding protocols  with respect to the desired properties of Section \ref{subsec:ledger}. In particular, we analyze \el, \mx, \ol, and \rc: we first provide a description of the protocols, then we abstract and identify their individual subprotocols mapping them to the sharding components of Section~\ref{sec:roadmap}. Using this mapping we either prove or disprove they meet our properties. We further discuss \cs , a well-established permissioned sharding protocol.
% The analysis can be found in Appendix~\ref{sec:eval}, while a summary can be found below and is illustrated in Table~\ref{tab:comparison}. 
%
We include in the evaluation the ``permissionless'' and ``slowly-adaptive'' properties to fairly compare the protocols. 
In our analysis, we evaluate the cross-shard communication protocols considering the fixes of~\cite{sonnino2019replay} against replay attacks.
% In the following, we present a summary of our evaluation for each protocol; we begin with the overview of each protocol and we include only the critical components to pinpoint where each protocol fails. 

\input{comparison-table.tex}

% We first show that \el does not satisfy consistency due to its \texttt{CrossShard} protocol, as long as the workload has any cross-shard transactions. 
We first show that \textbf{\el} does not satisfy \textit{consistency} in our model because the adversary may double-spend across shards when multi-input transactions are allowed (Theorem~\ref{thm:el-consistency}).
Additionally, \el does not satisfy \textit{scalability} by design regardless of the transaction distribution -- even with a few cross-shard transactions (Theorem~\ref{thm:el-scale}). Specifically,   
 all epoch-transition protocols are executed for every block while  parties maintain a global hash chain. Thus, transactions are only compressed  by a constant factor, the block size, resulting in space and communication growing proportionally to the number of parties. %the population
% We highlight scalability is not satisfied by design regardless of the transaction distribution, i.e., even with a few cross-shard transactions, because all parties maintain a global hash chain, and thus the space and communication factors are constant (equal to the block size).

We then show that \textbf{\mx} does not satisfy \textit{scalability} because miners must mine in parallel in all shards, verifying and storing all transactions to ensure security  (Theorem~\ref{thm:mx-scale}). Due to its design rationale, \mx cannot scale even with optimistic transaction distributions with no cross-shard transactions.

Third, we prove that \textbf{\ol} satisfies all properties but \textit{liveness} (Theorems~\ref{omni:persistence},\ref{omni:consistency},\ref{omni:liveness}, \ref{omni:scalability}). Specifically, \ol  checkpoints the UTXO pool at each epoch transition, but the state is not broadcasted to the network. Hence, a slowly adaptive adversary can corrupt a shard from the previous epoch before the new nodes of the shard bootstrap to the state in epoch transition. This attack violates liveness but simply adding a reliable broadcast step after checkpointing restores the liveness since all other components satisfy it already. The overhead of reliable broadcast can be amortized over the rounds of the epoch hence the overall scalability is not affected.

 Fourth, we prove \textbf{\rc} maintains a robust sharded ledger but only under a \textit{weaker model} than the one defined in Section~\ref{sec:model} (Theorems~\ref{rc:persistence},\ref{rc:consistency},\ref{rc:liveness},\ref{rc:scalability}).
Specifically, the protocol only allows a constant number of parties to join or leave and the adversary can at most corrupt a constant number of additional parties with each epoch transition. 
Another shortcoming of \rc is the synchronous consensus mechanism it employs. 
In case of temporary loss of synchrony in the network, the consensus of cross-shard transactions is vulnerable, hence consistency might break~\cite{zamani2018rapidchain}.
However, most of these drawbacks can be addressed with simple solutions, such as changing the consensus protocol (trade-off performance with security), replacing the epoch transition process with one similar to (fixed) \ol, etc. 
Although \ol (with the proposed fix) maintains a robust sharded ledger in a stronger model (as defined in Section~\ref{sec:model}), \rc introduces practical speedups on specific components of the system. These improvements are not asymptotically important -- and thus not captured by our framework -- but might be significant for the performance of  deployed sharding protocols.

Finally, we include in the comparison \textbf{\cs}, which maintains a robust sharded transaction ledger but only in the permissioned setting against a static adversary. 
\cs could be secure in our model in the permissioned setting if it adopts \ol's epoch transition protocols and the proposed fix for data availability in the verifiable compaction of state.
We omit the security proofs for \cs since they are either included in \cite{al2018chainspace} or are similar to \ol.
% We include in the evaluation the ``permissionless'' and ``slowly-adaptive'' property, in addition to our security and efficiency properties.

% \emph{Remark.} 
% Cross-shard communication protocols of some sharding systems suffer from replay attacks~\cite{sonnino2019replay}. In our analysis, we consider the fixes proposed in~\cite{sonnino2019replay}.
\vspace{-11pt}
\subsubsection*{Discussion.}
 Although we restrict our evaluation to the most impactful (so far) sharding proposals, we stress that the power of our framework and the bounds we provide are not limited to these works. For instance, we observe that Chainweb~\cite{martino2018chainweb}, a recently deployed sharding proposal, does not scale  because it violates Theorem~\ref{thm:scale-SPV}.
 We believe our framework is general enough to cover most sharding approaches, and we aspire it will be established as a tool for proving the security of future sharding protocols.

%% file: comparison-table.tex
\begin{table*}[ht]
\vspace{-0.3cm}
    \centering
 \caption{Summarizing sharding protocol properties under our model}
 \footnotesize{
\begin{tabular}{l|cccccc}
        \toprule
                     \text{Protocol}              & \text{Persistence} & \text{Consistency} & \text{Liveness} & \text{Scalability} & \text{Permissionless}  & \text{S.-adaptive}\\
        \midrule
                \text{\el}     &\cmark & \xmark & \cmark & \xmark & \cmark & \cmark\\

        \hline
                \text{\mx}     &\cmark & \cmark & \cmark & \xmark & \cmark & \cmark\\  
        \hline
                \text{\ol}         &\cmark & \cmark & \xmark & \cmark & \cmark  & \cmark \\ 
        \hline
                \text{\rc}      &\cmark & \cmark & \cmark & \cmark & \cmark & $\bm{\sim}$  \\ 
        \hline
                \text{\cs}     &\cmark & \cmark & \cmark & \cmark & \xmark & \xmark  \\ 
                
        \bottomrule
    \end{tabular}}
    \label{tab:comparison}
\end{table*}

%% file: 9figures.tex
\section*{Figures}\label{sec:figures}

\input{glossary.tex}

%% file: glossary.tex
% \section{Glossary}\label{glossary}
% ----- glossary ----
\begin{table}[h!]
    \centering
 \caption{(Glossary) The parameters in our analysis.}
\begin{tabular}{c|l}
        \toprule
    % \textbf{Variable name}              & \textbf{Variable meaning} \\

        % \hline
        $n$ & number of parties\\
        $f$ & number of Byzantine parties\\
%        $P_j$ & party $j$ \\
        $m$ & number of shards\\
        $v$ & average transaction size (number of inputs and outputs)\\
        $E$ & epoch, \ie, a set of consecutive rounds\\
        %% $UTXO$ & Unspent Transaction Output (data-structure for representing valid coins)\\
%        $tx$ & transaction\\
        $T$ & set of transactions (input)\\
%        $D_T$ & Distribution from which transaction inputs and transaction outputs are drawn (uniformly random in our analysis)\\
%      $S_i$ & set of possible valid ledgers of shard $i$ \\
%       $S$ & sharded transaction ledger\\
        %%it is the set of all $S_i$ \\
%	$x_i$ & set of transactions included in a block\\	
%	$L$ & transaction ledger (sequence of blocks)\\
%	$L_C$ & a ledger contained in chain $C$ \\
    %% 	$C$ & ?? \\ % I would remove this, it denoted the name of a chain 
%	$SL_j$ & Sharded Ledger of Party $j$  (set of ledgers, one ledger per shard that $P_j$ validates) \\
    $k$ & ``depth'' security parameter (persistence)\\
        %; indicates the number of blocks required to be build on top of the block including a transaction s.t. the transaction is considered stable.\\
        %The depth of a block which is part of a ledger \\
    $u$ & ``wait'' time (liveness)\\
        %%indicates the number of rounds required for a transaction to become stable since it was originally broadcast.   \\
%    $V(tx,L_i)$ & Verification oracle of transaction $tx$ against ledger $L_i$ \\
%    $M(\cdot,\cdot)$ & predicate that checks if two transactions  are conflicting \\
	% $\Sigma$ & scaling factor\\
	$\omega_m$ & communication  factor\\
	$\omega_s$ & space  factor\\
	$\omega_c$ & computational  factor\\
	$\sigma$ & throughput factor\\
	$\mu$ & chain quality parameter \\
	$\tau$ & chain growth parameter \\
	$v$ & average transaction size\\
	$m'$ & degree of parallelism\\
	$\gamma$ & average number of a shard's interacting shards (cross-shard)\\

        \bottomrule
    \end{tabular}
    \label{tab:glossary}
\end{table}

%% file: ack.tex
 \section*{Acknowledgments}
    The work was partially supported by the Austrian Science Fund (FWF) through the project CoRaF (grant agreement 2020388).

%% file: 3limit.tex
%----------------------------------------------------------------------
%-------------------------LOWER BOUNDS----------------------------------
%----------------------------------------------------------------------
\section{Limitations of sharding protocols}\label{app:limit}

\subsection{General Bounds}\label{subsec:general}
First, we prove there is no robust sharded transaction ledger that has a constant number of shards. Then, we show that there is no protocol that maintains a robust sharded transaction ledger against an adaptive adversary.

\begin{restatable}{lem}{shards}\label{thm:shards}
In any robust sharded transaction ledger the number of shards \\
(parametrized by $n$) is $m=\omega(1)$.
\end{restatable}
\begin{proof}
Suppose there is a protocol that maintains a constant number $m$ of sharded ledgers, denoted by $x_1, x_2, \dots, x_m$. Let $n$ denote the number of parties and $T$ the number of transactions to be processed (wlog assumed to be valid). 
A transaction is processed only if it is stable, \ie is included deep enough in a ledger ($k$ blocks from the end of the ledger where $k$ a security parameter).
Each ledger will include $T/m$ transactions on expectation. Now suppose each party participates in only one ledger (best case), thus broadcasts, verifies, and stores the transactions of that ledger only. Hence, every party stores $T/m$ transactions on expectation. The expected space  factor is $\omega_s=\sum_{\forall i \in [n]} \sum_{\forall x \in L_i} |x^{\lceil k}| / (n|T|) = \sum_{\forall x \in L_i} \frac{T}{nmT} = \frac{n}{nm} = \Theta (\frac{1}{m})=\Theta (1)$, when $m$ in constant. Thus, scalability is not satisfied.
\end{proof}

Suppose a party is participating in shard $x_i$. If the party maintains information (\eg\ the headers of the chain for verification purposes) on the chain of shard $x_j$, we say that the party is a \textit{light node} for shard $x_j$.
In particular, \textit{a light node for shard $x_j$ maintains information at least proportional to the length of the shard's chain $x_j$}. This holds because blocks must be of constant size to be able to scale in bandwidth (aka communication), and thus storing all the headers of a shard is asymptotically similar in overhead to storing the entire shard with the block content. 
Sublinear light clients~\cite{kiayias2020non,bunz2019flyclient}
%, such as NIPoPoW~\cite{kiayias2020non} and FlyClient~\cite{bunz2019flyclient},
verifiably compact the shard's state, thus are not considered light nodes but are discussed later.
 We next prove that if parties act as light clients to all shards involved in cross-shard transactions, then the sharded ledger can scale only if each shard does not interact with all the other shards (or a constant fraction thereof).

\begin{restatable}{lem}{crossbound}\label{lem:cross-bound} 
For any robust sharded transaction ledger that requires every participant to be a light node for all the shards affected by cross-shard transactions, it holds  $\mathbb{E}(\gamma)=o(m)$.
\end{restatable}
\begin{proof}
We assumed that every ledger interacts on average with $\gamma$ different ledgers, \ie, the cross-shard transactions involve $\gamma$ many different shards on expectation. 
The block size is considered constant, meaning each block includes at most $e$ transactions where $e$ is constant. Thus, each party maintaining a ledger and being a light node to $\gamma$ other ledgers must store on expectation $(1+\frac{\gamma}{e})\frac{T}{m}$ information. 
Hence, the expected space  factor is 
$$\mathbb{E}(\omega_s)= \sum_{\forall i \in [n]} \sum_{\forall x \in L_i} |x^{\lceil k}| / (n|T|) = n \frac{(1+\frac{\gamma}{e})\frac{T}{m}}{nT}= \Theta\Big(\frac{\gamma}{m}\Big)$$
where the second equation holds due to linearity of expectation.
To satisfy scalability, we demand  $\mathbb{E}(\omega_s)=o(1)$, thus $\gamma=o(m)$.%\hfill \qed
\end{proof}

Next, we show that there is no protocol that maintains a robust transaction ledger against an adaptive adversary in our model. 
We highlight that our result holds because we assume \textit{any node is corruptible} by the adversary.
If we assume more restrictive corruption sets, \eg\ each shard has at least one honest well-connected node, sharding against an adaptive adversary may be possible if we employ other tools, such as fraud and data availability proofs~\cite{albassam2018fraud}.

\adaptive*
\begin{proof} 
(Towards contradiction) Suppose there exists a protocol $\Pi$ that maintains a robust sharded ledger against an adaptive adversary that corrupts $f=n/m$ parties. 
From the pigeonhole principle, there exists at least one shard $x_i$ with at most $n/m$ parties (independent of how shards are created). The adversary is adaptive, hence at any round can corrupt all parties of shard $x_i$.
In a malicious shard, the adversary can perform arbitrary operations, thus can spend the same UTXO in multiple cross-shard transactions. However, for a cross-shard transaction to be executed it needs to be accepted by the output shard, which is honest. 
Now, suppose $\Pi$ allows the parties of each shard to verify the ledger of another shard. For Lemma \ref{lem:cross-bound} to hold, the verification process can affect at most $o(m)$ shards. Note that even a probabilistic verification, \ie, randomly select some transactions to verify, can fail due to storage requirements and the fact that the adversary can perform arbitrarily many attacks.
Therefore, for each shard, there are at least $2$ different shards that do not verify the cross-shard transactions (since Lemma \ref{lem:cross-bound} essentially states they cannot all be verified). Thus, the adversary can simply attempt to double-spend the same UTXO across every shard and will succeed in the shards that do not verify the validity of the cross-shard transaction. Hence, consistency is not satisfied.
%\hfill \qed
\end{proof}

%--------------------UAR BOUNDS-----------------------------------
\subsection{Bounds under Uniform Shard Creation}\label{subsec:random-shards}

In this section, we assume that the creation of shards is UTXO-dependent; transactions are assigned to shards independently and uniformly at random.
This assumption is in sync with the proposed protocols in the literature. In a non-randomized process of creating shards, the adversary can precompute and thus bias the process in a permissionless system. Hence, all sharding proposals employ a random process for shard creation. 
Furthermore, all shards validate approximately the same amount of transactions; otherwise the efficiency of the protocol would depend on the shard that validates most transactions.
For this reason, we assume the UTXO space is partitioned to shards uniformly at random. 
Note that we consider UTXOs to be random strings.

Under this assumption, we prove a constant fraction of transactions are cross-shard on expectation.
As a result, we prove no sharding protocol can maintain a robust sharded ledger when participants  act as  light  clients on all shards involved in cross-shard transactions.
Our observations hold for any transaction distribution $D_T$ that results in a constant fraction of cross-shard transactions.

\begin{restatable}{lem}{crossshardone}\label{lem:cross-shard}
The expected number of cross-shard transactions is $\Theta(|T|)$. 
\end{restatable}
\begin{proof}
Let $Y_i$ be the random variable that shows if a transaction is cross-shard; $Y_i=1$ if $tx_i \in T$ is cross-shard, and $0$ otherwise.
Since UTXOs are assigned to shards uniformly at random, $Pr[i\in x_k]=\frac{1}{m}$, for all $i\in v$ and $k \in [m]=\{1, 2,\dots, m \}$. 
The probability that all UTXOs in a transaction $tx\in T$ belong to the same shard is $\frac{1}{m^{v-1}}$ (where $v$ is the cardinality of UTXOs in $tx$).
%, due to the uniform distribution of UTXOs in transactions.
Hence, $Pr[Y_i=1]=1-\frac{1}{m^{v-1}}$. Thus, the expected number of cross-shard transactions is $\mathbb{E}(\sum_{\forall tx_i \in T} Y_i)= |T|\big(1-\frac{1}{m^{v-1}} \big)$. 
Since, $m(n)=\omega(1)$ (Lemma \ref{thm:shards}) and $v$ constant, the expected cross-shard transactions converges to $T$ for $n$ sufficiently large.
%\hfill \qed
\end{proof}

\begin{restatable}{lem}{crossshardtwo}\label{lem:cross-shard-s}
For any protocol that maintains a robust sharded transaction ledger, it holds $\gamma=\Theta(m)$.
\end{restatable}
\begin{proof} 
We assume each transaction has a single input and output, hence $v=2$. This is the worst-case input for evaluating how many shards interact per transaction; if $v\gg 2$ then each transaction would most probably involve more than two shards and thus each shard would interact with more different shards for the same set of transactions.

For $v=2$, we can reformulate the problem as a graph problem. 
Suppose we have a random graph $G$ with $m$ nodes, each representing a shard. 
Now let an edge between nodes $u$ and $w$ represent a  transaction between shards $u$ and $w$. Note that in this setting we allow self-loops, which represent the intra-shard transactions.
We create the graph $G$ with the following random process: We choose an edge independently and uniformly at random from the set of all possible edges including self-loops, denoted by $E'$. We repeat the process independently $|T|$ times, \ie, as many times as the cardinality of the transaction set.
We note that each trial is independent and the edges chosen uniformly at random due to the corresponding assumptions concerning the transaction set and the shard creation.
We will now show that the average degree of the graph is $\Theta(m)$, which immediately implies the statement of the lemma.

Let the random variable $Y_i$ represent the existence of edge $i$ in the graph, \ie, $Y_i=1$ if edge $i$ was created at any of the $T$ trials, $0$ otherwise. The set of all possible edges in the graph is $E$, $|E|=\binom{m}{2}=\frac{m(m-1)}{2}$. 
Note that this is not the same as set $E'$ which includes self-loops and thus $|E'|=\binom{m}{2}+m=\frac{m(m+1)}{2}$.
For any vertex $u$ of $G$, it holds \[\mathbb{E}[deg(u)]=\frac{2 \mathbb{E}[\sum_{\forall i \in E} Y_i]}{m} \]
where $deg(u)$ denotes the degree of node $u$.
We have, 
\[Pr[Y_i=1] = 1-Pr[Y_i=0]=\]\[ 1- Pr[Y_i=0 \text{ at trial 1}] Pr[Y_i=0 \text{ at trial 2}] \dots  \] \[Pr[Y_i=0 \text{ at trial T}]= 1-\Big(1-\frac{2}{m(m+1)} \Big)^{|T|}\]
Thus, \[\mathbb{E}[deg(u)]=\frac{2m(m-1)}{2}\Big[ 1-\Big(1-\frac{2}{m(m+1)} \Big)^{|T|} \Big]\] \[=(m-1)\Big[ 1-\Big(1-\frac{2}{m(m+1)} \Big)^{|T|} \Big]\]
Therefore, for many transactions we have $|T|=\omega(m^2)$ and consequently $\mathbb{E}[deg(u)]= \Theta(m)$.
%\hfill \qed
\end{proof}

\light*
\begin{proof}
Immediately follows from Lemmas \ref{lem:cross-bound} and  \ref{lem:cross-shard-s}.
%\hfill \qed
\end{proof}

\subsection{Bounds under Random Permutation of Parties to Shards}\label{subsec:random-parties}
In this section, we assume parties are periodically randomly shuffled among shards, using a random permutation of their IDs.
Any other shard assignment strategy yields equivalent or worse guarantees since we have no knowledge of which parties are Byzantine.
Our goal is to upper bound the number of shards for a protocol that maintains a robust sharded transaction ledger in our security model.
%, \ie, what is the maximum number of shards such that the security properties hold against an adversary that corrupts a constant fraction of parties $f=pn$.
To satisfy the security properties, we demand each shard to contain at least a constant fraction of honest parties $1-a$ $(< 1-\frac{f}{n})$, where $a$ is the tolerance of the shards. This is due to classic lower bounds of consensus protocols~\cite{lamport1982Byzantine}.

The \textit{size} of a shard is the number of the parties assigned to the shard.
We say shards are \textit{balanced} if all shards have approximately the same size. In what follows, we assume shards to be balanced (this can be done by drawing uniformly at random a balanced partition of parties).
We denote by $p=f/n$ the (constant) fraction of the Byzantine parties. A shard is \textit{a-honest} if at least a fraction of $1-a$ parties in the shard are honest.% otherwise the shard is \textit{a-malicious}.

The following lemma, proven by Raab and Steger~\cite{raab1998balls} will be useful later:
\begin{restatable}{lem}{shardsizelem}\label{lem:raab-steger}
Let M be the random variable that counts the number of balls in any bin if we throw $pn$ balls independently and uniformly at random into m bins. Then $Pr[M>k_{\alpha}] = o(1)$ if $\alpha >1$ and $Pr[M>k_{\alpha}] = 1 - o(1)$ if $0 <\alpha < 1$, where
\begin{equation} 
  k_{\alpha} =
    \begin{cases}
      \frac{\log{m}}{\log{\frac{m\log{m}}{pn}}} * (1+\alpha\frac{log^{(2)}\frac{m\log{m}}{pn}}{\log \frac{m\log{m}}{pn}}) & \text{if $\frac{m}{polylog(m)}\leq pn \ll m\log{m}$,}\\
      
      (d_c-1+\alpha)\log{m} & \text{if $pn=cm\log{m}$ for some constant $c$,}\\
      
      \frac{pn}{m} + \alpha \sqrt{2\frac{pn}{m}\log{m}} & \text{if $m\log{m}\ll pn\leq mpolylog(m)$,} \\
      
      \frac{pn}{m}+ \sqrt{\frac{2pn\log{m}}{m}(1-\frac{\log^{(2)}{m}}{2\alpha\log{m}})} & \textit{if $pn \gg m(\log{m})^{3}$}
    \end{cases}       
\end{equation}
\end{restatable}

\begin{restatable}{lem}{shardsizelem}\label{lem:shard-size}
Given $n$ parties are assigned uniformly at random to $m$ shards of constant size $s=\frac{n}{m}$ and the adversary corrupts at most $f=pn$ parties, all shards are $a$-honest ($p, a$ are constants with $p$ the proportion of corrupted parties and $a$ the tolerance of the model) with probability $1-o(1)$ if and only if the number of shards is at most $n=
cmlog(m)/p$, where $c$ is a constant and $p/a$ is small enough depending only on the value of $c$.% depending on $\alpha$ and $p$. % a or alpha ?
\end{restatable}
\begin{proof}
    We start by reformulating the problem in order to show it is equivalent to the well-know Generalized Birthday Paradox.

    Assuming we build $m$ shards of equivalent size $s=\frac{n}{m}$ using a random permutation with uniform probability. %Assume for simplicity that $s$ is an integer.
    Then this is equivalent to distributing the Byzantine processes to shards at random following a uniform law, but with the shards being of maximum size $s$. In other words, we throw $f=pn$ balls in $m$ bins of limited capacity $s$. We would like to know the probability that the maximum load of the bins be greater or equal to $a$. %=\frac{f}{m}
    
    Reformulated as the Birthday paradox, what is the probability that, in a room of $n$ people whose birthdays are spread uniformly at random over $m$ days, $a$ people share the same birthday? We denote that probability by $f(pn,m,a)$.
    
    Notice that our reformulation as the Birthday Paradox does not take into account the limited size of the possible birthdays (no more than $s$ people can have the same birthday). Both problems are however equivalent, as we can reconstruct that probability easily using Bayes' formula:
    
    $$P(A|B) = \frac{P(B|A)*P(A)}{P(B)} $$ 
    Where $A = $"the maximum load is $\leq as$", $B = $ "the maximum load is $\leq s$" and $A|B = C = $"all shards are $a-honest$.
    $P(B|A) = 1$ since $a<1$ so
    $$ P(C) = \frac{P(A)}{P(B)} $$ hence solving the Birthday Paradox solves our problem with very little additional calculation. Our calculation will actually be conducted using $A' = $"the maximum load is $\geq as$" and $B' = $ "the maximum load is $\geq s$"
    
    $$ P(C) = \frac{1-P(A')}{1-P(B')} $$

    Since $\frac{1-o(1)}{1-o(1)}\geq 1-o(1)$, it is sufficient for $P(C) = 1-o(1)$ that $P(A')=o(1)$ and $P(B') = o(1)$.
    %In order for the $C$ to be an almost-sure event, we need the probability of event $A'$ to be $o(1)$.
    The problem is sometimes denoted as the Cell Occupancy Problem~\cite{fisher2013birthday}.

    We then use lemma~\ref{lem:raab-steger} (beware, in the original paper~\cite{raab1998balls} $n$ and $m$ are reverse when compared with our notation).
    We want $\alpha > 1$, $k_{\alpha} = \frac{an}{m}$. 
    
    When applying this, we immediately get impossible equations for the third and fourth values of $k_{\alpha}$, hence it is not possible to have m in that range of values compared to n ($m\gg nlog(n)$) :

    $$ \frac{an}{m} = \frac{pn}{m} + \alpha \sqrt{2\frac{pn}{m}\log{m}}$$
    $$ \frac{(a-p)n}{m} = \alpha \sqrt{2\frac{pn}{m}\log{m}}$$
    $$ \frac{n}{m} = \frac{\alpha\sqrt{2p}}{(a-p)} \sqrt{\frac{n}{m}\log{m}}$$
    $$ \sqrt{n} = \frac{\alpha\sqrt{2p}}{(a-p)} \sqrt{m\log{m}}$$
    $$ n = \frac{\alpha^22p}{(a-p)^2}m\log{m}$$
    As we can see, we also violate the hypothesis that $pn\gg m\log{m}$, which is absurd. 
    %The same holds for $a=1$ (i.e. case $B'$).
    For the fourth equation, we can simply notice that since $\alpha >1$, $(1-\frac{\log^{(2)}{m}}{2\alpha\log{m}}) \leq 1$ hence reusing the calculation made for the third case $n$ will be even smaller when compared with $m\log{m}$, thus the hypothesis $pn \gg m(\log{m})^{3}$ is broken.
    
    The equations however is correct under the hypothesis that $pn = cm\log{m}$ (see calculation below). This indicates that this is as high a value of $m$ we can use while keeping the shards safe with overwhelming probability.
    $$ \frac{an}{m} = (d_c - 1 + \alpha)\log{m} $$
    $$ n = \frac{1}{a}(d_c - 1 + \alpha)m\log{m} $$
    We can see already that we are indeed verifying the hypothesis $pn = cm\log{m}$ for some constant $c$ (the constant $d_c$ is a scalar not dependant on either $n$ or $m$). 
    if $k_{\alpha} = \frac{n}{m}$, then $n = (d_c - 1 + \alpha)m\log{m}$ and the hypothesis is also verified.

    We now need to make sure that $\alpha > 1$ for both cases. 

    Since, by hypothesis, $pn = cm\log{m}$, we identify that $c = \frac{p}{a}(d_c - 1 + \alpha)$, where $d_c \geq c$. In order to obtain $\alpha > 1$, it is necessary that $c > \frac{p}{a} d_c$ where $p < a$. $d_c$ is a function of $c$ with $d_c >c$, hence for a given $c$ it is always possible to enforce $\alpha > 1$ if $p/a$ is small enough.
    
    for the case $k_{\alpha} = \frac{n}{m}$, the previous result holds trivially with $a = 1$.

    % conclusion of the proof
    
\end{proof}

Using the previous calculations, we can exhibit the trade-off between security and scalability in a mathematical formulation in corollary~\ref{cor:shard-size}. A systems designer may choose to adjust either parameter $p/a$ or $c$, one being computed thanks to the chosen value of the other. Since the expression is not mathematically intuitive, we provide a plotting of the increasing function $p/a = g(c)$ in Figure~\ref{fig:tradeoff}.

% 1+\frac{x}{y}\left(\ln\left(x\right)-\ln\left(\frac{x}{y}\right)+1\right)-x=0

\begin{figure}[ht!]
    \centering
    \begin{subfigure}{.5\textwidth}
    \includegraphics[scale=0.121,trim={0 1.9cm 0 1.7cm},clip]{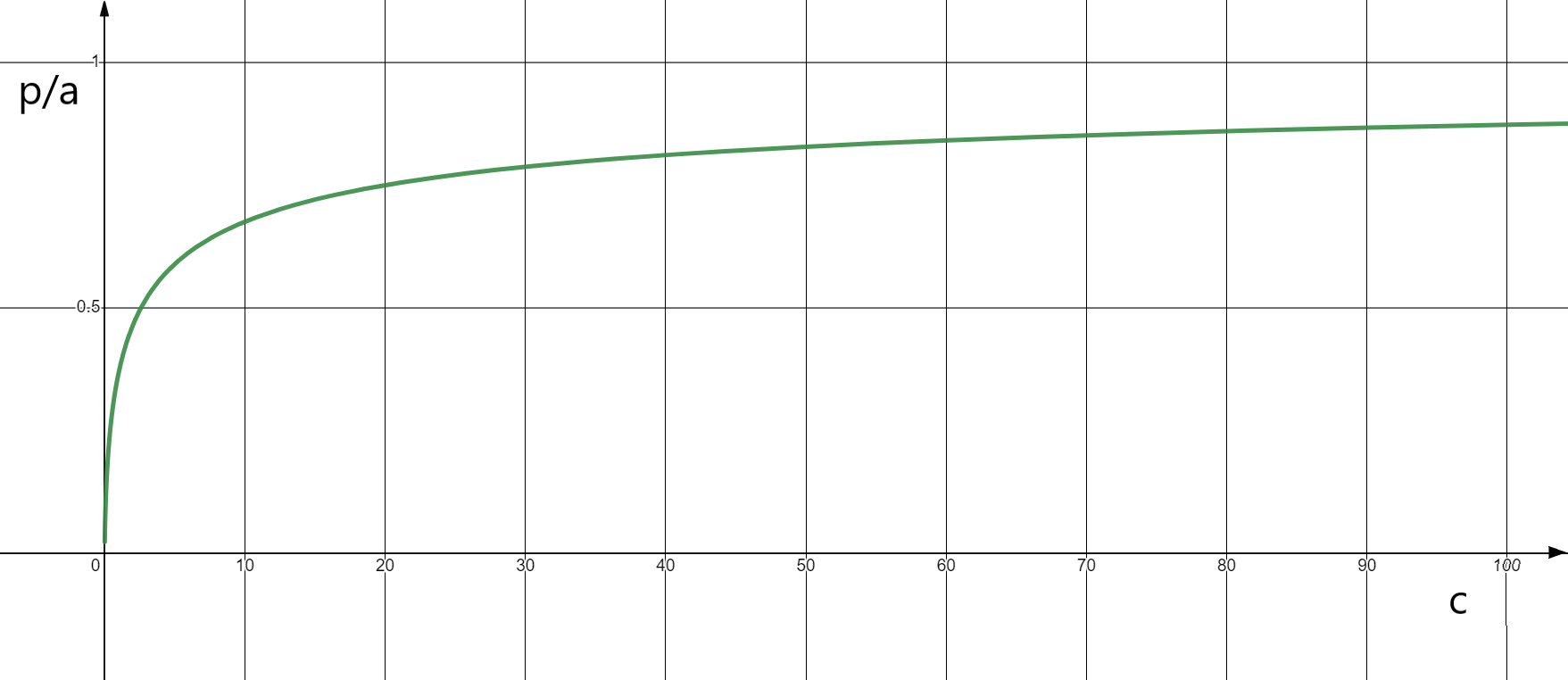}
    \caption{zoomed from 0 to 100}
    \label{fig:subsmall}
\end{subfigure}%
\begin{subfigure}{.5\textwidth}
    \centering
    \includegraphics[scale=0.13]{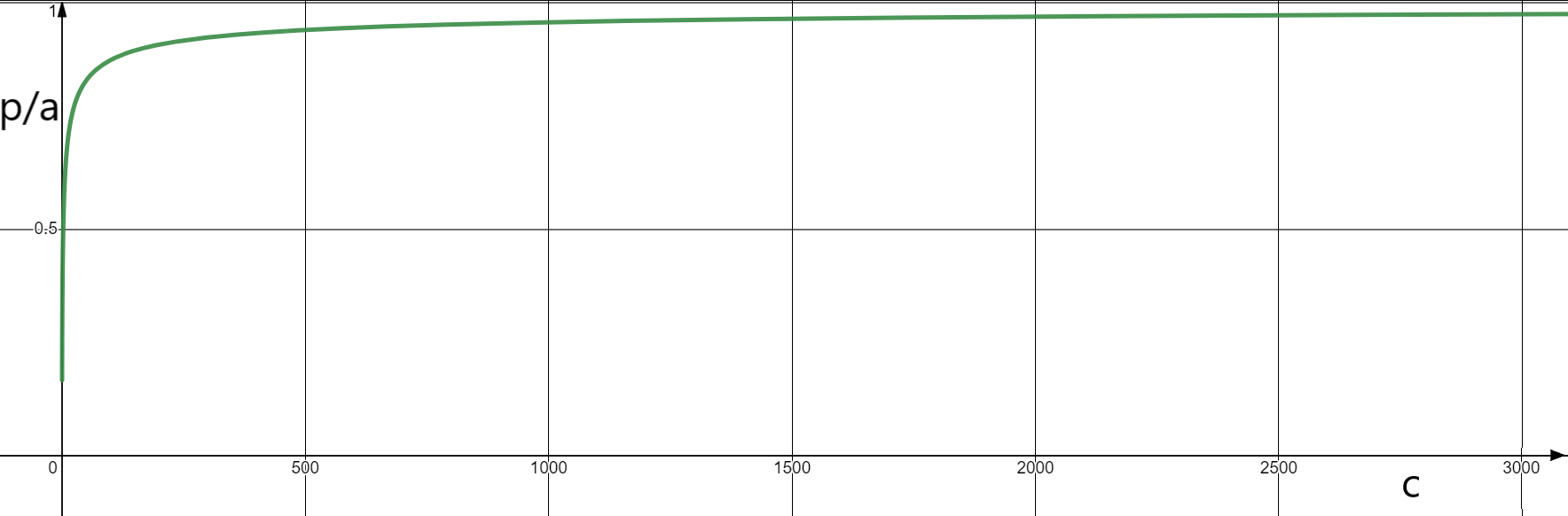}
    \caption{zoomed from 0 to 3000}
    \label{fig:sub3000}
    \end{subfigure}
    \caption{$p/a = g(c)$ as described in corollary~\ref{cor:shard-size}. $p$ is the  proportion of corrupted parties in the system, while $1-a$ is the maximum proportion of corrupted parties allowed per shard.}
    \label{fig:tradeoff}
\end{figure}

\begin{restatable}{cor}{shardsizecor}\label{cor:shard-size}
    In a sharding protocol maintaining a robust sharded transaction ledger against an adversary, the trade-off between scalability (low value of $c$) and security (high value of $p/a$) is described by $\frac{c}{d_c} > \frac{p}{a}$. $c$ is the multiplicative constant in the relation $pn=cm\log(m)$, $d_c$ is a function of $c$, while $p$  and $1-a$ are the proportion of corrupted parties in the system and per shard, respectively.
\end{restatable}
\begin{proof}
According to lemma~\ref{lem:shard-size}, the constant $d_c$ is a real number dependant only on $c$ and 
$$\frac{c}{d_c} > \frac{p}{a}$$ which means the value of $p/a$ is ceiled by the value of $c/d_c$.

As explained in~\cite{raab1998balls}, $d_c$ is the solution to the equation $1 + x(log(c) - log(x) + 1 ) - c = 0$ that is greater than $c$. Thus we have the exact mathematical expression of the well-known security/scalability trade-off.
\end{proof}

\begin{restatable}{cor}{upper_bound_cor}\label{cor:upper-bound}
    In a sharding protocol maintaining a robust sharded transaction ledger against an adversary, $m$ is upper-bounded by $f(n) = \frac{n}{c'\log(\frac{n}{c'\log(n)})}$ with $c' = \frac{c}{p}$ and $c$ a constant as described in corollary~\ref{cor:shard-size}.
\end{restatable}
\begin{proof}
    Because of lemma~\ref{lem:shard-size}, $cm\log(m)=pn$. using $m = \frac{n}{c'\log(m)}$ (a), we obtain $m = \frac{n}{c'\log(\frac{n}{c'\log(m)})}$ and since $ n \geq m$, an upper-bound is $f(n) = \frac{n}{c'\log(\frac{n}{c'\log(n)})}$. Note we could build a tighter but more complex upper bound by replacing $m$ by its expression (a) instead of $n$ as many times as desired. 
\end{proof}

Next, we prove that any sharding protocol may scale at most by an $n/\log{n}$ factor.
This bound refers to independent nodes. If, for instance, we ``shard'' per authority, but all authorities represented in each shard, the bound of the theorem does not hold and the actual system should be considered sharded since every authority holds all the data.

\scale*
\begin{proof}
In our security model, the adversary can corrupt $f=pn$ parties, $p$ constant. Hence, from Corollary \ref{cor:shard-size}, $m = O(\frac{n}{\log m})$. Each party stores at least $T/m$ transactions on average and thus the expected space  factor is $\omega_s \geq n\frac{T/m}{T}=\frac{n}{m}$. Therefore, any sharding protocol can scale at most $O(\frac{n}{\log m})$.
\end{proof}

Next, we show that any sharding protocol that satisfies scalability requires some process of \textit{verifiable compaction of state} such as checkpoints~\cite{kokoris2017omniledger}, cryptographic accumulators~\cite{boneh2019batching}, zero-knowledge proofs~\cite{zkproofs2020}, non-interactive proofs of proofs-of-work~\cite{kiayias2020non,bunz2019flyclient}, proof of necessary work~\cite{assimakis2019proof} or erasure codes~\cite{kadhe2019sef}.
Such a process allows the state of the distributed ledger (\eg, stable transactions) to be compressed significantly while users can verify the correctness of the state.
Intuitively, in any  sharding protocol secure against a slowly adaptive adversary parties must periodically shuffle in shards. To verify new transactions the parties must receive a verifiably correct UTXO pool for the new shard without downloading the full shard history; otherwise the communication overhead of the bootstrapping process eventually exceeds that of a non-sharded blockchain.
Although existing evaluations typically ignore this aspect with respect to bandwidth, we stress its importance in the long-term operation: \textit{the bootstrap cost will eventually become the bottleneck due to the need for nodes to regularly shuffle.}

\checkpoints*
\begin{proof}
(Towards contradiction) 
Suppose there is a protocol that maintains a robust sharded ledger without employing any process that verifiably compacts the blockchain. 
To guarantee security against a slowly-adaptive adversary, the parties change shards at the end of each epoch. 
At the beginning of each epoch, the parties must process a new set of transactions. To check the validity of this new set of transactions, each (honest) shard member downloads and maintains the corresponding ledger. 
Note that even if the party only maintains the hash-chain of a ledger, the cost is equivalent to maintaining the list of transactions given that the block size is constant.
We will show that the communication  factor increases with time, eventually exceeding that of a non-sharded blockchain; thus scalability is not satisfied from that point on.

% \todo{check}
In each epoch transition, a party changes shards with probability $1-1/m$, where $m$ is the number of shards. 
As a result, a party changing a shard in epoch $k$ must download the shard's ledger of size $\dfrac{k\cdot T}{m}$. Therefore, the expected communication  factor of bootstrapping during the $k$-th epoch transition is $\dfrac{k\cdot T}{m} \cdot (1-\dfrac{1}{m})$. 
We observe the communication overhead grows with the number of epochs $k$, hence it will eventually become the scaling bottleneck. For instance, for $k> m\cdot n$, the communication factor is greater than linear to the number of parties in the system $n$, thus the protocol does not satisfy scalability.
\end{proof}
 Theorem~\ref{thm:checkpoints} holds even if parties are not assigned to shards uniformly at random but follow some other shuffling strategy like in \cite{rana2020free2shard}. \textit{As long as a significant fraction of honest parties change shards from epoch to epoch, verifiable compaction of state is necessary} to restrict the bandwidth requirements during bootstrapping in order to satisfy scalability.

%% file: 4analysis.tex
\section{Analysis}\label{app:analysis}
We show that Divide \& Scale is secure in our model (\ie, satisfies persistence, consistency, and liveness), while its efficiency (\ie, scalability and throughput factor) depends on the chosen subprotocols.
For the purpose of our analysis, we assume all employed subprotocols satisfy liveness.

\begin{thrm}
Divide \& Scale satisfies persistence in our system model assuming at most f Byzantine nodes.
\end{thrm}
\begin{proof}
Assuming \texttt{Sybil} guarantees the fair distribution of identities (Sybil, property iv), and \texttt{Divide2Shard} maintains the distribution within the desired limits to guarantee the securities bounds of \texttt{Consensus} (Divide2Shard, property iii), the common prefix property is satisfied in each shard, so persistence is satisfied.
% Assuming \texttt{Divide2Shard} respects the security bounds of \texttt{Consensus}, the common prefix property is satisfied in each shard, so persistence is satisfied.
\end{proof}

\begin{thrm}
Divide \& Scale satisfies consistency in our system model assuming at most f Byzantine nodes.
\end{thrm}
\begin{proof}
Transactions can either be intra-shard (all UTXOs within a single shard) or cross-shard.
Consistency is satisfied for intra-shard transactions as long as \texttt{Sybil} and \texttt{Divide2Shard} result in a distribution that respects the security bounds of \texttt{Consensus}, hence the common prefix property is satisfied.
Furthermore, consistency is satisfied for cross-shard transactions from the \texttt{CrossShard} protocol as long as it correctly provides atomicity.
\end{proof}

\begin{thrm}
Divide \& Scale satisfies liveness in our system model assuming at most f Byzantine nodes.
\end{thrm}
\begin{proof}
Follows from the assumption that all subprotocols satisfy liveness, as well as the \texttt{CompactState} protocol that ensures data availability between epochs.
\end{proof}

\subsubsection*{{Scalability.}}
The scalability of Divide \& Scale depends on the worse scaling factor, \ie, communication, space, computation, of all the components it employs. 
The maximum scaling factor for \texttt{DRG, Divide2Shards, Sybil}, and \texttt{CompactState} can be amortized over the rounds of an epoch because these protocols are executed once per epoch. Thus, the size of an epoch is critical for scalability.
Intuitively, this implies that \textit{if the size of the epoch is small, hence the adversary highly-adaptive, sharding is not that beneficial as the protocols that are executed on the epoch transaction are as resource demanding as the consensus in a non-sharded system.}
% \begin{cor}
% The size of an epoch $R=\omega (XXX)$.
% \end{cor}
% \begin{proof}
% The DRG protocols in literature have communication complexity at least XXX.
% A sharding protocol satisfies scalability when the maximum overhead factor is $o(n)$, therefore the amortized communication complexity of the DRG protocol must be more than XXX. We deduce that the number of rounds in an epoch must be at least $XXX$.
% \end{proof}

\subsubsection*{{Throughput factor.}}
Similarly to scalability, the throughput factor also depends on the chosen subroutines, and in particular,  \texttt{Consensus} and \texttt{CrossShards}.
To be specific, the throughput factor depends on the shard growth and shard quality parameters which are determined by  \texttt{Consensus}. 
In addition, given a transaction input, the degree of parallelism, which is the last component of the throughput factor, is determined by the maximum number of shards possible and the way cross-shard transactions are handled. The maximum number of shards depends on  \texttt{Consensus} and  \texttt{Divide2Shards}, while  \texttt{CrossShard} determines how many shards are affected by a single transaction.
For instance, if the transactions are divided in shards uniformly at random, Divide \& Scale can scale at most by $n/\log n$ as stated in Corollary~\ref{cor:shard-size}. We further note that the minimum number of affected shards for a specific transaction is the number of UTXOs that map to different shards; otherwise security cannot be guaranteed.

We demonstrate in Appendix~\ref{sec:eval} how to calculate the scaling factors and the throughput factor for \ol and \rc.

%% file: 8eval-new.tex
%------------------------------------------------------------------------
%-------------------------EVALUATION--------------------------------------
%----------------------------------------------------------------------
\section{Evaluation of Existing Protocols}\label{sec:eval}
In this section, we evaluate existing sharding protocols in our model with respect to the desired properties defined in Section \ref{subsec:ledger}. 
A summary of our evaluation can be found in Table~\ref{tab:comparison} in Section~\ref{app:comaprison}.

The analysis is conducted in the synchronous model and thus any details regarding performance on periods of asynchrony are discarded. The same holds for other practical refinements that do not asymptotically improve performance.

%------------------------ELASTICO------------------------
\subsection{Elastico}
\subsubsection{Overview.}
\el is the first distributed blockchain sharding protocol introduced by Luu et al.~\cite{luu2016secure}. The protocol lies in the intersection of traditional BFT protocols and the Nakamoto consensus. 
The protocol is synchronous and proceeds in epochs. 
The setting is permissionless, and during each epoch, the participants create valid identities for the next epoch by producing proof-of-work (PoW) solutions. 
The adversary is slowly-adaptive (see Section~\ref{sec:model}) and controls at most $25\%$ of the  computational power of the system or equivalently $f< \frac{n}{4}$ out of $n$ valid identities in total.

At the beginning of each epoch, parties are partitioned into small shards (committees) of constant size $c$. The number of shards is $m=2^s$, where $s$ is a small constant such that $n=c \cdot 2^s$. A shard member contacts its directory committee to identify the other members of the same shard.
For each party, the directory committee consists of the first $c$ identities created in the epoch in the party's local view.
Transactions are randomly partitioned in disjoint sets based on the hash of the transaction input (in the UTXO model); hence, each shard only processes a fraction of the total transactions in the system. The shard members execute a BFT protocol to validate the shard's transactions and then send the validated transactions to the final committee. The final committee consists of all members with a fixed $s$-bit shards identity, and is in charge of two operations: (i) computing and broadcasting the final block, which is a digital signature on the union of all valid received transactions\footnote{The final committee in \el broadcasts only the Merkle root for each block. However, this is asymptotically equivalent to including all transactions since the block size is constant. Furthermore, the final committee does not check if the received transactions are conflicting but merely verifies the presence of signatures.} (via executing a BFT protocol), and (ii) generating and broadcasting a bounded exponential biased random string to be used as a public source of randomness in the next epoch (\eg\ for the PoW).

\paragraph{\tt{Consensus:}} 
\el does not specify the consensus protocol but instead can employ any standard BFT protocol, like PBFT \cite{castro1999practical}.

\paragraph{\tt{CrossShard \& StatePartition:}} 
Each transaction is assigned to a shard according to the hash of the transaction's inputs. Every party maintains the entire blockchain, thus each shard can validate the assigned transaction independently, \ie, there are no cross-shard transactions. 
Note that \el assumes that transactions have a single input and output, which is not the case in cryptocurrencies as discussed in Section~\ref{sec:limit}. To generalize \el's transaction assignment method to multiple inputs, we assume each transaction is assigned to the shard corresponding to the hash of all its inputs. Otherwise, if each input is assigned to a different shard according to its hash value, an additional protocol is required to guarantee the atomicity of transactions and hence the security (consistency) of \el.

\paragraph{\tt{Sybil:}}
Participants create valid identities by producing PoW solutions using the randomness of the previous epoch.

\paragraph{\tt{Divide2Shards \& CompactState:}} 
The protocol assigns each identity to a random shard in $2^s$, identified by an $s$-bit shard identity.
At the end of each epoch, the final committee broadcasts the final block that contains the Merkle hash root of every block of all shards' block. The final block is stored by all parties in the system. 
Hence, when the parties are re-assigned to new shards they already have the hash-chain to confirm the shard ledger and future transactions.
Essentially, an epoch in \el is equivalent to a block generation round.

\paragraph{\tt{DRG:}} 
In each epoch, the final committee (of size $c$) generates a set of random strings $R$ via a commit-and-XOR protocol. First, all committee members generate an $r$-bit random string $r_i$ and send the hash $h(r_i)$ to all other committee members. Then, the committee runs an interactive consistency protocol to agree on a single set of hash values $S$, which they include on the final block. Later, each (honest) committee member broadcasts its random string $r_i$ to all parties in the network. Each party chooses and XORs $c/2+1$ random strings for which the corresponding hash exists in $S$. The output string is the party's randomness for the epoch. Note that $r>2\lambda+c-\log(c)/2$, where $\lambda$ is a security parameter.

\subsubsection{Analysis.}
\el's threat model allows for adversaries that can drop or modify messages, and send different messages to honest parties, which is not allowed in our model. 
However, we show that even under a more restrictive adversarial model, \el fails to meet the desired sharding properties. Specifically, we prove \el does not satisfy \textit{scalability} and \textit{consistency}.
From the security analysis of~\cite{luu2016secure}, it follows that \el satisfies persistence and liveness in our system model.

% \begin{thrm}
% \el satisfies persistence.
% \end{thrm}
% \begin{proof}
% Follows from Lemma 4~\cite{luu2016secure}.
% \end{proof}

% \begin{thrm}
% \el satisfies liveness.
% \end{thrm}
% \begin{proof}
% During each epoch liveness follows from the BFT protocol of each shard since no cross-shard communication is necessary (Lemma 4~\cite{luu2016secure}). 
% Hence, we only need to show liveness is satisfied during epoch transition. 
% \end{proof}

\begin{thrm}\label{thm:el-consistency}
\el does not satisfy consistency in our system model.
\end{thrm}
\begin{proof}
Suppose a party submits two valid transactions, one spending input $x$ and another spending input $x$ and input $y$. Note that the second is a single transaction with two inputs. In this case, the probability that both hashes (transactions), $H(x,y)$ and $H(x)$, land in the same shard is $1/m$. Hence, the probability of a successful double-spending in a set of $T$ transactions is almost $1-(1/m)^T$, which converges to $0$ as $T$ grows, for any value $m>1$.
However, $m>1$ is necessary to satisfy scalability (Lemma \ref{thm:shards}).
Therefore, there will be almost surely a round in which two parties report two conflicting transactions. Since the final committee does not verify the validity of transactions but only checks the appropriate signatures are present, consistency is not satisfied.
\end{proof}

\begin{lem}\label{lem:elastico-ssf}
The communication and space factors of \el are $\omega_m=\Theta(n)$ and $\omega_s=\Theta(1)$.
\end{lem}
\begin{proof}
At the end of each epoch, which corresponds to the generation of one block per shard, the final committee broadcasts the final block to the entire network. All parties download and store the final block.
hence all parties maintain the entire input set of transactions. 
Since the block size is considered constant, downloading and storing the final block which consists of the hash-chains of all shards is equivalent to downloading and storing all the shards' ledgers. 
It follows that the space factor is $\omega_s=\Theta(1)$ as all parties store a constantly-compressed version of the input $T$, regardless of the nature of the input set $T$.
Similarly, it follows that the communication factor is $\omega_m=\Theta(n)$ as the broadcast of the final block takes place regularly at the generation of one block per shard, \ie, \el's epoch.\hfill \qed
\end{proof}

\begin{thrm}\label{thm:el-scale}
\el does not satisfy scalability in our system model.
\end{thrm}
\begin{proof}
Immediately follows from Definition \ref{def:scalability} and 
Lemma \ref{lem:elastico-ssf}.\hfill \qed
\end{proof}

%------------------------MONOXIDE------------------------
\subsection{Monoxide}

\subsubsection{Overview.} 
Monoxide \cite{wang2019monoxide} is an asynchronous proof-of-work protocol, where the adversary controls at most $50\%$ of the computational power of the system.
The protocol uniformly partitions the space of user addresses into shards (zones) according to the first $k$ bits.
Every party is permanently assigned to a shard uniformly at random. Each shard employs the GHOST \cite{sompolinsky2015ghost} consensus protocol.

Participants are either full-nodes that verify and maintain the transaction ledgers, or miners investing computational power to solve PoW puzzles for profit in addition to being full-nodes. 
\mx introduces a new mining algorithm, called Chu-ko-nu, that enables miners to mine in parallel for all shards. 
The Chu-ko-nu algorithm aims to distribute the hashing power to protect individual shards from an adversarial takeover.
Successful miners include transactions in blocks. 
A block in \mx is divided into two parts: the chaining block that includes all metadata (Merkle root, nonce for PoW, etc.) creating the hash-chain, and the transaction-block that includes the list of transactions. All parties maintain the hash-chain of every shard in the system.

Furthermore, all parties maintain a distributed hash table for peer discovery and identifying parties in a specific shard. This way the parties of the same shard can identify each other and cross-shard transactions are sent directly to the destination shard.
Cross-shard transactions are validated in the shard of the payer and verified from the shard of the payee via a relay transaction and the hash-chain of the payer's shard.

\paragraph{\tt{Consensus:}} 
The consensus protocol of each shard is GHOST~\cite{sompolinsky2015ghost}. GHOST is a DAG-based consensus protocol similar to Nakamoto consensus~\cite{nakamoto2008bitcoin}, but the consensus selection rule is the heaviest subtree instead of the longest chain.

\paragraph{\tt{StatePartition:}} 
\mx is account-based hence all transactions are single input and single output.

\paragraph{\tt{CrossShard:}} 
An input shard is a shard that corresponds to the address of a sender of a transaction (payer) while an output shard one that corresponds to the address of a receiver of a transaction (payee).
Each cross-shard transaction is processed in the input shard, where an additional relay transaction is created and included in a block. The relay transaction consists of all metadata needed to verify the validity of the original transaction by only maintaining the hash-chain of a shard (\ie for light nodes).
The miner of the output shard verifies that the relay transaction is stable and then includes it in a block in the output shard. 
Note that in case of forks in the input shard, \mx invalidates the relay transactions and rewrites the affected transaction ledger to maintain consistency.

\paragraph{\tt{Sybil:}}
In a typical PoW election scheme, the adversary can create many identities and target its computational power to specific shards to gain control over more than half of the shard's participants. In such a case, the security of the protocol fails (both persistence and consistency properties do not hold). 
To address this issue, \mx introduces a new mining algorithm, Chu-ko-nu, that allows parallel mining on all shards.
Specifically, a miner can batch valid transactions from all shards and use the root of the Merkle tree of the list of chaining headers in the batch as input to the hash, alongside with the nonce (and some configuration data). Thus, when a miner successfully computes a hash lower than the target, the miner adds a block to every shard.

\paragraph{\tt{Divide2Shards:}} 
Parties are permanently  assigned to shards uniformly at random according the first $k$ bits of their address.

\paragraph{\tt{DRG:}} 
The protocol uses deterministic randomness (\eg\ hash function) and does not require any random source.

\paragraph{\tt{CompactState:}} 
No compaction of state is used in \mx.

% \paragraph{Epoch transition.}
% The protocol is asynchronous (there are no epochs), and parties are permanently assigned to shards.

\subsubsection{Analysis.}
We prove that \mx satisfies persistence, liveness, and consistency, but \textit{does not satisfy scalability}. 
The same result is also immediately derived from our impossibility result stated in Theorem ~\ref{thm:scale-SPV} as \mx demands each party to verify cross-shard transactions by acting as a light node to all shards; effectively demonstrating the effectiveness of our framework and the usability of our results.

\begin{restatable}{thrm}{mxperslive}
\mx satisfies persistence and liveness in our system model for $f<n/2$.
\end{restatable}
\begin{proof}
From the analysis of \mx, it holds that if all honest miners follow the Chu-ko-nu mining algorithm, then honest majority within each shard holds with high probability for any adversary with $f<n/2$ (Section 5.3~\cite{wang2019monoxide}).

Assuming honest majority within shards, persistence depends on two factors: the probability a stable transaction becomes invalid in a shard's ledger, and the probability a cross-shard transaction is reverted after being confirmed. Both these factors solely depend on the common prefix property of the shards' consensus mechanism.
\mx employs GHOST as the consensus mechanism of each shard, hence the common prefix property is satisfied if we assume that invalidating the relay transaction does not affect other shards~\cite{kiayias2017trees}. Suppose common prefix is satisfied with probability $1-p$ (which is overwhelming on the ``depth'' security parameter $k$). Then, the probability none of the outputs of a transaction are invalidated is $(1-p)^{(v-1)}$ (worst case where $v-1$ outputs -- relay transactions -- link to one input). Thus, a transaction is valid in a shard's ledger after $k$ blocks with probability $(1-p)^v$, which is overwhelming in $k$ since $v$ is considered constant. Therefore, persistence is satisfied.

Similarly, liveness is satisfied within each shard. 
Furthermore, this implies liveness is satisfied for cross-shard transactions. In particular, both the initiative and relay transactions will be eventually included in the shards' transaction ledgers, as long as chain quality and chain growth are guaranteed within each shard~\cite{kiayias2017trees}.\hfill \qed
\end{proof}

\begin{restatable}{thrm}{mxconsistency}\label{mx:consistency}
\mx satisfies consistency in our system model for $f<n/2$.
\end{restatable}
\begin{proof}
The common prefix property is satisfied in GHOST~\cite{kiayias2017ouroboros}  with high probability. Thus, intra-shard transactions satisfy consistency with high probability (on the ``depth'' security parameter).
Furthermore, if a cross-shard transaction output is invalidated after its confirmation, \mx allows rewriting the affected transaction ledgers. Hence, consistency is restored in case of cross-transaction failure.
Thus, overall, consistency is satisfied in \mx.\hfill \qed
\end{proof}

Note that allowing to rewrite the transaction ledgers in case a relay transaction is invalidated strengthens the consistency property but weakens the persistence and liveness properties. 

Intuitively, to satisfy persistence in a sharded PoW system, the adversarial power needs to be distributed across shards. To that end, \mx employs a new mining algorithm, Chu-ko-nu, that incentivizes honest parties to mine in parallel on all shards. However, this implies that a miner needs to verify transactions on all shards and  maintain a transaction ledger for all shards. Hence, the computation and space factors are proportional to the number of (honest) participants and the protocol does not satisfy scalability.

\begin{thrm}\label{mx:scalability}\label{thm:mx-scale}
\mx does not satisfy scalability in our system model for $f<n/2$.
\end{thrm}

\begin{proof}
Let $m$ denote the number of shards (zones), $m_p$ the fraction of mining power running the Chu-ko-nu mining algorithm and $m_d$ the rest of the mining power ($m_p+m_d=1$). Additionally, suppose $m_s$ denotes the mining power of one shard. 
The Chu-ko-nu algorithm enforces the parties to verify transactions that belong to all shards, hence the parties store all sharded ledgers. 
To satisfy scalability, the space factor of \mx can be at most $o(1)$. Similarly, it follows that the verification overhead expressed through the computational factor must be bounded by $o(n)$. Thus, at most $o(n)$ parties can run the Chu-ko-nu mining algorithm, hence $n m_p=o(n)$.
We note that the adversary will not participate in the Chu-ko-nu mining algorithm as distributing the hashing power is to the adversary's disadvantage.  
% Hence, $n m_d=\omega(1)$.

To satisfy persistence, every shard running the GHOST protocol~\cite{sompolinsky2015ghost} must satisfy the common prefix property. Thus, the adversary cannot control more than $m_a < m_s/2$ hash power, where $m_s=\frac{m_d}{m}+m_p$. Consequently, we have $m_a < \frac{m_s}{2(m_d+m_p)} = \frac{1}{2}- \frac{m_d(m-1)}{2m(m_d+m_p)}$. For $n$ sufficiently large, $m_p$ converges to $0$; hence $m_a < \frac{1}{2}- \frac{(m-1)}{2m} = \frac{1}{2m}$. 
From Lemma \ref{thm:shards}, $m=\omega(1)$, thus the adversarial power $m_a < 0$ for sufficiently large $n$. 
We conclude that \mx does not satisfy scalability in our model. Moreover, we identify in \mx a clear trade-off  between security and scaling storage and verification. \hfill \qed
\end{proof}

%------------------------OMNILEDGER------------------------
\subsection{OmniLedger}

\subsubsection{Overview.}
OmniLedger \cite{kokoris2017omniledger} proceeds in epochs, assumes a partially synchronous model within each epoch (to be responsive), synchronous communication channels between honest parties (with a large maximum delay), and a slowly-adaptive computationally-bounded adversary that can corrupt up to $f<n/4$ parties.

The protocol bootstraps using techniques from ByzCoin~\cite{kogias2016enhancing}. The core idea is that there is a global identity blockchain that is extended once per epoch with Sybil resistant proofs (proof-of-work, proof-of-stake, or proof-of-personhood~\cite{borge2017pop}) coupled with public keys. 
At the beginning of each epoch a sliding window mechanism is employed to define the eligible validators as the ones with identities in the last $W$ blocks, where $W$ depends on the adaptivity of the adversary. For our definition of slowly adaptive, we set $W=1$.
The UTXO space is partitioned uniformly at random into $m$ shards, each shard maintaining its own ledger.

At the beginning of each epoch, a new common random value is created via a distributed randomness generation (DRG) protocol. The DRG protocol employs verifiable random functions (VRF) to elect a leader who runs RandHound \cite{syta2017scalable} to create the random value.
The random value is used as a challenge for the next epoch's identity registration and as a seed to assigning identities of the current epoch into shards.

Once the participants for this epoch are assigned to shards and bootstrap their internal states, they start validating transactions and updating the shards' transaction ledgers by operating ByzCoinX, a modification of ByzCoin~\cite{kogias2016enhancing}.
When a transaction is cross-shard, a protocol that ensures the atomic operation of transactions across shards called \textit{Atomix} is employed. Atomix is a client-driven atomic commit protocol secure against Byzantine adversaries. 

\paragraph{\tt{Consensus:}}
\ol suggests the use of a strongly consistent consensus in order to support Atomix. This modular approach means that any consensus protocol~\cite{castro1999practical,kogias2016enhancing,ren2017practical,gilad2017algorand,kokoris2019robust} works with \ol as long as the deployment setting of \ol respects the limitations of the consensus protocol. 
In its experimental deployment, \ol uses a variant of ByzCoin~\cite{kogias2016enhancing} called ByzCoinX~\cite{kokoris2019robust} in order to maintain the scalability of ByzCoin and be robust as well. We omit the details of ByzCoinX as it is not relevant to our analysis.

\paragraph{\tt{StatePartition:}} 
The UTXO space is partitioned uniformly at random into $m$ shards.

\paragraph{\tt{CrossShard (Atomix):}} 
Atomix is a client-based adaptation of two-phase atomic commit protocol running with the assumption that the underlying shards are correct and never crash. This assumption is satisfied because of the random assignment of parties to shards, as well as the Byzantine fault-tolerant consensus of each shard. 

In particular, Atomix works in two steps: First, the client that wants the transaction to go through requests a proof-of-acceptance or proof-of-rejection from the shards managing the inputs, who log the transactions in their internal blockchain. Afterwards, the client either collects proof-of-acceptance from all the shards or at least one proof-of-rejection. In the first case, the client communicates the proofs to the output shards, who verify the proofs and finish the transaction by generating the necessary UTXOs. In the second case, the client communicates the proofs to the input shards who revert their state and abort the transaction. 
Atomix, has a subtle replay attack, hence we analyze \ol with the proposed fix~\cite{sonnino2019replay}.

\paragraph{\tt{Sybil:}}
A global identity blockchain with Sybil resistant proofs coupled with public keys is extended once per epoch. 

\paragraph{\tt{Divide2Shards:}} 
Once the parties generate the epoch randomness, the parties can independently compute the shard they are assigned to for this epoch by permuting ($mod$ n) the list of validators (available in the identity chain).

\paragraph{\tt{DRG:}} 
The DRG protocol consists of two steps to produce unbiasable randomness. On the first step, all parties evaluate a VRF using their private key and the randomness of the previous round to generate a ``lottery ticket''. Then the parties broadcast their ticket and wait for $\Delta$ to be sure that they receive the ticket with the lowest value whose generator is elected as the leader of RandHound. 

This second step is a partially-synchronous randomness generation protocol, meaning that even in the presence of asynchrony safety is not violated. If the leader is honest, then eventually the parties will output an unbiasable random value, whereas if the leader is dishonest there are no liveness guarantees.
To recover from this type of fault the parties can view-change the leader and go back to the first step in order to elect a new leader. 

This composition of randomness generation protocols (leader election and multiparty generation) guarantees that all parties agree on the final randomness (due to the view-change) and the protocol remains safe in asynchrony. Furthermore, if the assumed synchrony bound (which can be increasing like PBFT~\cite{castro1999practical}) is correct, an honest leader will be elected in a constant number of rounds.

Note, however, that the DRG protocol is modular, thus any other scalable distributed randomness generation protocol with similar guarantees, such as Hydrand~\cite{schindler2018hydrand} or Scrape~\cite{cascudo2017scrape}, can be used.

\paragraph{\tt{CompactState:}} 
A key component that enables \ol to scale is the epoch transition. At the end of every epoch, the parties run consensus on the state changes and append the new state (\eg\ UTXO pool) in a state-block that points directly to the previous epoch's state-block. This is a classic technique~\cite{castro1999practical} during reconfiguration events of state machine replication algorithms called checkpointing. 
New validators do not replay the actual shard's ledger but instead, look only at the checkpoints which help them bootstrap faster. 

In order to guarantee the continuous operation of the system, after the parties finish the state commitment process, the shards are reconfigured in small batches (at most $1/3$ of the parties in each shard at a time). If there are any blocks committed after the state-block, the validators replay the state-transitions directly.

\subsubsection{Analysis.}
In this section, we prove \ol satisfies persistence, consistency, and scalability (on expectation) but \textit{fails to satisfy liveness}. Nevertheless, we estimate the efficiency of \ol by providing an upper bound on its throughput factor.

\begin{restatable}{lem}{omnirandomsecure}\label{omni:random-secure}
At the beginning of each epoch, \ol provides an unbiased, unpredictable, common to all parties random value (with overwhelming probability in $t$ within $t$ rounds). 
\end{restatable}
\begin{proof}
If the elected leader that orchestrates the distributed randomness generation protocol (RandHound or equivalent) is honest the statement holds. 
On the other hand, if the leader is Byzantine, the leader cannot affect the security of the protocol, meaning the leader cannot bias the random value. However, a Byzantine leader can delay the process by being unresponsive. We show that there will be an honest leader, hence the protocol will output a random value, with overwhelming probability in the number of rounds $t$.

The adversary cannot pre-mine PoW puzzles, because the randomness of each epoch is used in the PoW calculation of the next epoch. Hence, the expected number of identities the adversary will control (number of Byzantine parties) in the next epoch is $f<n/4$. Hence, the adversary will have the smallest ticket -- output of the VRF -- and thus will be the leader that orchestrates the distributed randomness generation protocol (RandHound) with probability $1/2$.
Then, the probability there will be an honest leader in $t$ rounds is $1-\frac{1}{2^t}$, which is overwhelming in $t$. 

The unpredictability is inherited by the properties of the employed distributed randomness generation protocol.
\hfill \qed
\end{proof}

\begin{restatable}{lem}{omnirandomscale}\label{omni:random-scale-lemma}
The distributed randomness generation protocol has  $O(\frac{n\log^2 n}{R})$ amortized communication complexity, where $R$ is the number of rounds in an epoch.
\end{restatable}
\begin{proof} 
The DRG protocol inherits the communication complexity of RandHound, which is $O(c^2n)$ \cite{schindler2018hydrand}. 
In \cite{syta2017scalable}, the authors claim that $c$ is constant. However, the protocol requires a constant fraction of honest parties (\eg\ $n/3$) in each of the $n/c$ partitions of size $c$ against an adversary that can corrupt a constant fraction of the total number of parties (\eg\ $n/4$). 
Hence, from Lemma \ref{lem:shard-size}, we have $c=\Omega(\log n)$, which leads to communication complexity $O(n\log^2n)$ for each epoch.
Assuming each epoch consist of $R$ rounds, the amortized per round communication complexity is $O(\frac{n\log^2 n}{R})$.
\hfill \qed
\end{proof}

\begin{restatable}{cor}{omnishardsize}\label{omni:shard-size}
In each epoch, the expected size of each shard is $n/m$.
\end{restatable}
\begin{proof}
Due to Lemma \ref{omni:random-secure}, the $n$ parties are assigned independently and uniformly at random to $m$ shards. Hence, the expected number of parties in a shard is $n/m$.
\hfill \qed
\end{proof}

\begin{restatable}{lem}{omnihonestshards}\label{omni:honest-shards}
In each epoch, all shards are $\frac{1}{3}$-honest for $m\leq f(n)$ with $f(n)$ as described in corollary~\ref{cor:upper-bound}.
%In each epoch, all shards are $\frac{1}{3}$-honest for $m\leq \frac{n}{300c \log n}$, where $c$ is a security parameter.
\end{restatable}
\begin{proof}
Due to Lemma \ref{omni:random-secure}, the $n$ parties are assigned independently and uniformly at random to $m$ shards. Since $a=1/3 > p=1/4$, both $a,p$ constant, the statement holds from Lemma~\ref{lem:shard-size} 
and corollary~\ref{cor:upper-bound}.
%for $m=\frac{n}{c' \log n}$, where $c' > 300c$.
\hfill \qed
\end{proof}
Note that the bound is theoretical and holds for a  large number of parties since the probability tends to $1$ as the number of parties grows. For practical bounds, we refer to \ol's analysis~\cite{kokoris2017omniledger}.

\begin{restatable}{thrm}{omnipersistence}\label{omni:persistence}
\ol satisfies persistence in our system model for $f<n/4$.
\end{restatable}
\begin{proof}
From Lemma \ref{omni:honest-shards}, each shard has an honest supermajority $\frac{2}{3}\frac{n}{m}$ of participants. Hence, persistence holds by the common prefix property of the consensus protocol of each shard. Specifically, for ByzCoinX, persistence holds for depth parameter $k=1$ because ByzCoinX guarantees finality. 
\hfill \qed
\end{proof}

\begin{thrm}\label{omni:liveness}
\ol does not satisfy liveness in our system model for $f<n/4$.
\end{thrm}
\begin{proof}
To estimate the liveness of the protocol, we need to examine  all the subprotocols: (i) {\tt Consensus}, (ii) {\tt CrossShard} or Atomix, (iii) {\tt DRG}, (iv) {\tt CompactState}, and (v) {\tt Divide2Shards}.

{\tt Consensus:} From Lemma \ref{omni:honest-shards}, each shard has an honest supermajority $\frac{2}{3}\frac{n}{m}$ of participants. Hence, in this stage liveness holds by chain growth and chain quality properties of the underlying blockchain protocol (an elaborate proof can be found in \cite{garay2015bitcoin}). 
The same holds for {\tt CompactState} as it is executed similarly to {\tt Consensus}.

{\tt CrossShard:} Atomix guarantees liveness since the protocol's efficiency depends on the consensus of each shard involved in the cross-shard transaction.
Note that liveness does not depend on the client's behavior; if the appropriate information or some part of the transaction is not provided in multiple rounds to the parties of the protocol then the liveness property does not guarantee the inclusion of the transaction in the ledger. Furthermore, if some other party wants to continue the process it can collect all necessary information from the ledgers of the shards.

{\tt DRG:} During the epoch transition, the DRG protocol provides a common random value with overwhelming probability within $t$ rounds (Lemma \ref{omni:random-secure}). Hence, liveness is satisfied in this subprotocol as well.

{\tt Divide2Shrds:} Liveness is not satisfied in this protocol. The reason is that a slowly-adaptive adversary can select who to corrupt during epoch transition, and thus can corrupt a shard from the previous epoch. Since the compact state has not been disseminated in the network, the adversary can simply delete the shard's state. Thereafter, the data unavailability prevents the progress of the system. 
\hfill \qed
\end{proof}

\begin{thrm}\label{omni:consistency}
\ol satisfies consistency in our system model for $f<n/4$.
\end{thrm}
\begin{proof}
Each shard is $\frac{1}{3}$-honest (Lemma \ref{omni:honest-shards}). Hence, consistency holds within each shard, and the adversary cannot successfully double-spend. 
Nevertheless, we need to guarantee consistency even when transactions are cross-shard. \ol employs Atomix, a protocol that guarantees cross-shard transactions are atomic. Thus, the adversary cannot validate two conflicting transactions across different shards. 

Moreover, the adversary cannot revert the chain of a shard and double-spend an input of a cross-shard transaction after the transaction is accepted in all relevant shards because persistence holds (Theorem \ref{omni:persistence}). Suppose persistence holds with probability $p$. Then, the probability the adversary breaks consistency in a cross-shard transaction is the probability of successfully double-spending in one of the relevant to the transaction shards, $1-p^v$, where $v$ is the average size of transactions. Since $v$ is constant, consistency holds with high probability, given that persistence holds with high probability.
\hfill \qed
\end{proof}

To prove \ol satisfies scalability (on expectation) we need to evaluate the scaling factors in the following subprotocols of the system: (i) {\tt Consensus}, (ii) {\tt CrossShard}, (iii) {\tt DRG}, and (iv) {\tt Divide2Shards}.
Note that {\tt CompactState} is merely an execution of {\tt Consensus}.

\begin{restatable}{lem}{omnishardscale}\label{omni:shard-scale}
The scaling factors of {Consensus} are $\omega_m=O({n}/{m})$, $\omega_s=O({1}/{m})$, and $\omega_c=O({n}/{m})$.
\end{restatable}
\begin{proof}
From Corollary \ref{omni:shard-size}, the expected number of parties in a shard is $n/m$. ByzCoin has quadratic to the number of parties' worst-case communication complexity, hence  the communication  factor of the protocol is $O(n/m)$. 
The verification complexity collapses to the communication complexity. 
The space factor is $O({1}/{m})$, as each party maintains the ledger of the assigned shard for the epoch. \hfill \qed
\end{proof}

\begin{restatable}{lem}{omniatomixscale}\label{omni:atomix-scale}
The communication factor of Atomix ({CrossShard}) is $\omega_m=O(v\frac{n}{m})$, where $v$ is the average size of transactions.
\end{restatable}
\begin{proof}
In a cross-shard transaction, Atomix allows the participants of the output shards to verify the validity of the transaction's inputs without maintaining any information on the input shards' ledgers. This holds due to persistence (see Theorem \ref{omni:persistence}). 

Furthermore, the verification process requires each input shard to verify the validity of the transaction's inputs and produce a proof-of-acceptance or proof-of-rejection. This corresponds to one query to the verification oracle for each input.
In addition, each party of an output shard must verify that all proofs-of-acceptance are present and no shard rejected an input of the cross-shard transaction. 
The proof-of-acceptance (or rejection) consists of the signature of the shard which is linear to the number of parties in the shard. 
The relevant parties have to receive all the information related to the transaction from the client (or leader), hence the communication factor is $O(v\frac{n}{m})$.
% Similarly, the verification complexity is $O(v\frac{n}{m})$, because each party of the input shards has to make a query to the verification oracle (worst case where the number of inputs is approximately $v$).
% is of constant size (signature of the shard) for each input shard and the parties have to receive all the information related to the transaction from the client (or leader), hence the verification complexity collapses to the communication complexity. The communication/verification overhead factor at this step is $O(v^2)$ (worst case where the number of inputs is approximately $v$).

So far, we considered the communication complexity of Atomix. However, each input must be verified within the corresponding input shard. From Lemma \ref{omni:shard-scale}, we get that the communication factor at this step is $O(v\frac{n}{m})$. 
\end{proof}

\begin{restatable}{lem}{omniepochscale}\label{omni:epoch-scale}
The communication factor of {Divide2Shards} is $\omega_m=O(\frac{n}{mR})$, while the space factor is $\omega_s=O(1/R)$, where $R$ is size of an epoch. 
\end{restatable}
\begin{proof}
During the epoch transition each party is assigned to a shard uniformly at random and thus most probably needs to bootstrap to a new shard, meaning the party must store the new shard's ledger. At this point, within each shard \ol introduces checkpoints, the state blocks that summarize the state of the ledger ({CompactState}). Therefore, when a party syncs with a shard's ledger, it does not download and store the entire ledger but only the active UTXO pool corresponding to the previous epoch's state block.

For security reasons, each party that is reassigned to a new shard must receive the state block of the new shard by $O(n/m)$ parties. Thus, the communication complexity of the protocol is $O(\frac{n}{mR})$ amortized per round, where $R$ is the number of rounds in an epoch.
% Hence, for $R=\Omega(\log n)$ number of rounds in each epoch, the expected communication overhead factor (amortized per round) is constant. 

The space complexity is constant but amortized over the epoch length since the state block has a constant size and is broadcast once per epoch, $\omega_s=O(1/R)$. There is no verification process at this stage.

% Since the stored data is at least as much as the data each party downloads during an epoch transition, the space overhead factor dominates this stage.

% Let us calculate the space overhead factor after $e$ epochs.
% The input transaction set is $T=\sum^{e}_{i=1}T_i$. 
% For any party of the protocol the expected storage cost at each epoch $i$ is $T_i/m$. Hence, at epoch $e$, the expected storage cost of a party is $\sum^{e}_{i=1}\frac{T_i}{m}$. Summing over all parties and normalizing over the input, we have the space overhead factor $\frac{n}{T}\sum^{e}_{i=1}\frac{T_i}{m}= \frac{nT}{mT}=\frac{n}{m}$.
% Note that the space overhead factor remains the same even if we take into account the slow mitigation of parties proposed in \ol to maintain operability during the epoch transition. This is because each party can maintain at most two ledgers in parallel at this stage, which does not affect the space overhead factor asymptotically.

% There is no verification process during this stage. Furthermore, each party that is reassigned in a new shard must receive the state block of the new shard by $O(n/m)$ parties for security reasons. Thus, the communication complexity of the protocol is $O(\frac{n\log n}{R})$ amortized per round, where $R$ is the number of rounds in an epoch.
% Hence, for $R=\Omega(\log n)$ number of rounds in each epoch, the expected communication overhead factor (amortized per round) is constant. 
% Therefore, the space overhead factor is dominant during this phase.
\hfill \qed
\end{proof}

\begin{thrm}\label{omni:scalability}
\ol satisfies scalability in our system model for  $f<n/4$  with communication and computational factor $O(n/m)$ and space factor $O(1/m)$, where $n=O(m \log m)$.
\end{thrm}
\begin{proof}
To evaluate the scalability of \ol, we need to estimate the dominating scaling factors of all the subprotocols of the system:
(i) {\tt Consensus}, (ii) {\tt CrossShard}, (iii) {\tt DRG}, and (iv) {\tt Divide2Shards}.

The scaling factors of {Consensus} are $\omega_m=O({n}/{m})$, $\omega_s=O({1}/{m})$, and $\omega_c=O({n}/{m})$ (Lemma~\ref{omni:shard-scale}),  while Atomix ({\tt CrossShard}) has expected communication factor $O(v\frac{n}{m})$ (Lemma \ref{omni:atomix-scale}) where the average size of transaction $v$ is constant (see Section \ref{sec:limit}).

The epoch transition consists of the {\tt DRG}, {\tt CompactState}, and {\tt Divide2Shards} protocols.
We assume a large enough epoch in rounds, $R=\Omega(n \log n)$, in order to amortize the communication-heavy protocols that are executed only once per epoch. 
{\tt CompactState} has the same overhead as Consensus hence it is not critical.
For $R=\Omega(n \log n)$, {\tt DRG} has an expected amortized communication factor $O(\log n)$ (Lemma \ref{omni:random-scale-lemma}),  while {\tt Divide2Shards} has an expected amortized communication factor of $\omega_m=O(\frac{1}{m \log n})$  and an amortized space factor of $\omega_s=O(1/R)= O(\frac{1}{n \log n})$(Lemma \ref{omni:epoch-scale}). 

Overall, considering the worst of the aforementioned scaling factors for \ol, we have expected communication and computational factors $O(n/m)$ and space factor $O(1/m)$, where $n=O(m \log m)$  (see Lemma~\ref{thm:shards} and Lemma~\ref{omni:honest-shards}).
\hfill \qed
\end{proof}

\begin{restatable}{thrm}{omnithroughput}\label{omni:throughput}
In \ol, the throughput factor is $\sigma=\mu \cdot \tau \cdot \dfrac{m}{v} < \frac{\mu \cdot \tau \cdot f(n)}{v} $ where $f(n) = \frac{n}{c'\log(\frac{n}{c'\log(n)})}$ with $c' = \frac{c}{p}$ and $c$ a constant as described in corollary~\ref{cor:shard-size}.
%< \mu \cdot \tau \cdot \dfrac{n}{\log n} \cdot \dfrac{(a-p)^2}{2+a-p} \cdot \dfrac{1}{v}$.\todo{to be changed-Antoine}
\end{restatable}
\begin{proof}
In Atomix, at most $v$ shards are affected per transaction, thus $m'<m/v$ \footnote{Note that if $v$ is constant, a more elaborate analysis could yield a lower upper bound on $m'$ better than $m/v$ (depending on $D_T$). However, if $v$ is not constant but approximates the number of shards $m$, then $m'$ is also bounded by the scalability of the Atomix protocol (Lemma \ref{omni:atomix-scale}), and thus the throughput factor can be much lower.}.
From Lemma \ref{lem:shard-size} and corollary~\ref{cor:upper-bound}, 
$n \leq f(n)$.
Therefore, $\sigma < \frac{\mu \cdot \tau \cdot f(n)}{v}$
%$m<\dfrac{n}{\log n} \cdot \dfrac{(a-p)^2}{2+a-p}$.
%Therefore, $\sigma < \mu \cdot \tau \cdot \dfrac{n}{\log n} \cdot \dfrac{(a-p)^2}{2+a-p} \cdot \dfrac{1}{v}$
\hfill \qed
\end{proof}

% \begin{rem}
The parameter $v$ depends on the input transaction set.
The parameters $\mu, \tau, a, p $ depend on the choice of the consensus protocol. 
Specifically, $\mu$ represents the ratio of honest blocks in the chain of a shard. On the other hand, $\tau$ depends on the latency of the consensus protocol, \ie, what is the ratio between the propagation time and the block generation time. Last, $a$ expresses the resilience of the consensus protocol (\eg, $1/3$ for PBFT), while $p$ the fraction of corrupted parties in the system ($f=pn$).
% \end{rem}

In \ol, the consensus protocol is modular, so we chose to maintain the parameters for a fairer comparison to other protocols.

% To provide an example on these parameters and evaluate the throughput of \ol, we use Motor \cite{kokoris2019robust}. Motor is an improvement on ByzCoin, which we cannot use because it does not provide any guarantees for chain quality. Motor, on the other hand, guarantees chain quality $\mu=2/3$ for $a=1/3$ resilience within each shard. The chain growth of the protocol is $\tau=1/4$ since $4$ rounds are needed to commit a block on the blockchain.
% Furthermore, we assume $v=5$, as indicated from real data \cite{bitcoinvisuals}.
% Thus, \ol, with $p=1/4$ and Motor as the consensus protocol, maintains a robust transaction ledger in our security model with throughput factor \[\sigma < \frac{2}{3} \cdot \frac{1}{4} \cdot \frac{n}{\log n} \cdot \frac{(\frac{1}{3}-\frac{1}{4})^2}{2+\frac{1}{3}-\frac{1}{4}} \cdot \frac{1}{5} =\frac{n}{9000 \log n}  \]

%-----------------------RAPIDCHAIN-------------------------------
\subsection{RapidChain}
% \vspace{-5pt}
\subsubsection{Overview.}
RapidChain \cite{zamani2018rapidchain} is a synchronous protocol and proceeds in epochs. The adversary is slowly-adaptive, computationally-bounded and corrupts less than $1/3$ of the participants ($f<n/3$).

The protocol bootstraps via a committee election protocol that selects $O(\sqrt{n})$ parties -- the root group. The root group generates and distributes a sequence of random bits used to establish the reference committee. The reference committee consists of $O(\log n)$ parties, is re-elected at the end of each epoch, and is responsible for: (i) generating the randomness of the next epoch, (ii) validating the identities of participants for the next epoch from the PoW puzzle, and (iii) reconfiguring the shards from one epoch to the next (to protect against single shard takeover attacks).

The parties are divided into shards of size $O(\log n)$ (committees). Each shard handles a fraction of the transactions, assigned based on the prefix of the transaction ID. Transactions are sent by external users to an arbitrary number of active (for this epoch) parties. The parties then use an inter-shard routing scheme (based on Kademlia~\cite{maymounkov2002kademlia}) to send the transactions to the input and output shards, \ie, the shards handling the inputs and outputs of a transaction, resp.

To process cross-shard transactions, the leader of the output shard creates an additional transaction for every different input shard. Then the leader sends (via the inter-shard routing scheme) these transactions to the corresponding input shards for validation. 
% To process cross-shard transactions, the leader of the output shard creates one additional transaction for every different input shard which the leader sends (via the inter-shard routing scheme) to the input shard for validation. 
To validate transactions (\ie, a block), each shard runs a variant of the synchronous consensus of Ren et al.~\cite{ren2017practical} and thus tolerates $1/2$ Byzantine parties. 

At the end of each epoch, the shards are reconfigured according to the participants registered in the new reference block. Specifically, \rc uses a bounded version of Cuckoo rule \cite{sen2012commensal}; the reconfiguration protocol adds a new party to a shard uniformly at random, and also moves a constant number of parties from each shard and assigns them to other shards uniformly at random. 
% \vspace{-5pt}
\paragraph{\tt{Consensus:}} 
In each round, each shard randomly picks a leader. The leader creates a block, gossips the block header $H$ (containing the round and the Merkle root) to the members of the shard, and initiates the consensus protocol on $H$. 
The consensus protocol consists of four rounds: (1) The leader gossips $(H,propose)$, (2) All parties gossip the received header $(H,echo)$, (3) The honest parties that received at least two echoes containing a different header gossip $(H',pending)$, where $H'$ contains the null Merkle root and the round, (4) Upon receiving $\frac{nf}{m}+1$ echos of the same and only header, an honest party gossips $(H,accept)$ along with the received echoes.
To increase the transaction throughput, \rc allows new leaders to propose new blocks even if the previous block is not yet accepted by all honest parties.
% (pipeline transactions).
% \vspace{-5pt}

\paragraph{\tt{StatePartition:}} 
Each shard handles a fraction of the transactions, assigned based on the prefix of the transaction ID.

\paragraph{\tt{CrossShard:}} 
For each cross-shard transaction, the leader of the output shard creates one ``dummy'' transaction for each input UTXO in order to move the transactions' inputs to the output shard, and execute the transaction within the shard. 
To be specific, assume we have a transaction with two inputs $I_1,I_2$ and one output $O$. 
% Note that all outputs are assigned to the same shard according to the transaction ID.
The leader of the output shard creates three new transactions: $tx_1$ with input $I_1$ and output $I'_1$, where $I'_1$ holds the same amount of money with $I_1$ and belongs to the output shard. $tx_2$ is created similarly. $tx_3$ with inputs $I'_1$ and $I'_2$ and output $O$. Then the leader sends $tx_1, tx_2$ to the input shards respectively. In principle, the output shard is claiming to be a trusted channel~\cite{androulaki2018channels} (which is guaranteed from the assignment), hence the input shards should transfer their assets there and then execute the transaction atomically inside the output shard (or abort by returning their assets back to the input shards).
\vspace{-5pt}
\paragraph{\tt{Sybil:}}
A party can only participate in an epoch if it solves a PoW puzzle with the previous epoch's randomness, submit the solution to the reference committee, and consequently be included in the next reference block. The reference block contains the active parties' identities for the next epoch, their shard assignment, and the next epoch's randomness, and is broadcast by the reference committee at the end of each epoch.
% \vspace{-5pt}
\paragraph{\tt{Divide2Shards:}} 
During bootstrapping, the parties are partitioned independently and uniformly at random in groups of size $O(\sqrt{n})$ with a deterministic random process. Then, each group runs the DRG protocol and creates a (local) random seed. Every node in the group computes the hash of the random seed and its public key. The $e$ (small constant) smallest tickets are elected from each group and gossiped to the other groups, along with at least half the signatures of the group. These elected parties are the root group. The root group then selects the reference committee of size $O(\log n)$, which in turn partitions the parties randomly into shards as follows: 
each party is mapped to a random position in $[0,1)$ using a hash function. Then, the range $[0,1)$ is partitioned into $k$ regions, where $k$ is constant. A shard is the group of parties assigned to $O(\log n)$ regions.

During epoch transition, a constant number of parties can join (or leave) the system. This process is handled by the reference committee which determines the next epoch's shard assignment, given the set of active parties for the epoch. The reference committee divides the shards into two groups based on each shard's  number of active parties in the previous epoch: group $A$ contains the $m/2$ larger in size shards, while the rest comprise group $I$.
Every new node is assigned uniformly at random to a shard in $A$. Then, a constant number of parties is  evicted from each shard and assigned uniformly at random in a shard in $I$.
% \vspace{-5pt}
\paragraph{\tt{DRG:}}
\rc uses Feldman's verifiable secret sharing \cite{feldman1987practical} to distributively generate unbiased randomness. At the end of each epoch, the reference committee executes a distributed randomness generation (DRG) protocol to provide the random seed of the next epoch. The same DRG protocol is also executed during bootstrapping to create the root group. 
\vspace{-5pt}
\paragraph{\tt{CompactState:}} 
No protocol for compaction of the state is used.

\subsubsection{Analysis.}
\rc does not maintain a robust sharded transaction ledger under our security model since it assumes a weaker adversary. To fairly evaluate the protocol, we weaken our security model. First, assume the adversary cannot change more than a constant number of Byzantine parties during an epoch transition, which we term \textit{constant-adaptive adversary}. 
In general, we assume \textit{bounded epoch transitions}, \ie, at most a constant number of leave/join requests during each transition. Furthermore, the number of epochs is asymptotically less than polynomial to the number of parties.
In this weaker security model, we prove \rc maintains a robust sharded transaction ledger, and provide an upper bound on the throughput factor of the protocol.

Note that in cross-shard transactions, the ``dummy'' transactions that are committed in the shards' ledgers as valid, spend UTXOs that are not signed by the corresponding users. Instead, the original transaction, signed by the users, is provided to the shards to verify the validity of the ``dummy'' transactions.
Hence, the transaction validation rules change. 
Furthermore, the protocol that handles cross-shard transactions has no proof of security against Byzantine leaders. 
For analysis purposes, we assume the following holds:
\begin{assumption}\label{rc:cross-claim}
CrossShard satisfies safety even under a Byzantine leader (of the output shard). 
%This is necessary only for the liveness of the transaction.
\end{assumption}
% Assumption~\ref{rc:cross-claim} is necessary only for the liveness of the transaction.\todo{why? its is used in the proof of consistency...}

\begin{restatable}{lem}{rcrandomscale}\label{rc:random-scale}
The communication factor of DRG is $O(n/m)$.
\end{restatable}
\begin{proof} 
The DRG protocol is executed by the final committee once each epoch. The size of the final committee is $O(n/m)=O(\log n)$.
The communication complexity of the DRG protocol is quadratic to the number of parties \cite{feldman1987practical}. Thus, the communication factor is $O(n/m)$.
\hfill \qed
\end{proof}

\begin{restatable}{lem}{rchonestshards}\label{rc:honest-shards}
In each epoch, all shards are $\frac{1}{2}$-honest for $m\leq f(n)$ with $f(n)$ from corollary~\ref{cor:upper-bound}.
%$m\leq \frac{n}{78c \log n}$, where $c$ is a security parameter.
\end{restatable}
\begin{proof}
During the bootstrapping process of \rc (first epoch), the $n$ parties are partitioned independently and uniformly at random into $m$ shards \cite{feldman1987practical}. For  $p=1/3$, the shards are $\frac{1}{2}$-honest only if 
$m\leq f(n)$ with $f(n)$ from corollary~\ref{cor:upper-bound}.
% $m\leq \frac{n}{78c \log n}$, where $c$ is a security parameter (Lemma \ref{lem:shard-size}).
At any time during the protocol, all shards remain $\frac{1}{2}$-honest (\cite{zamani2018rapidchain}, Theorem $5$). Hence, the statement  holds after each epoch transition, as long as the number of epochs is $o(n)$.
\hfill \qed
\end{proof}

\begin{restatable}{lem}{rcshardsize}\label{rc:shard-size}
In each epoch, the expected size of each shard is $O(n/m)$.
\end{restatable}
\begin{proof}
During the bootstrapping process of \rc (first epoch), the $n$ parties are partitioned independently and uniformly at random into $m$ shards \cite{feldman1987practical}. The expected shard size in the first epoch is $n/m$. 
Furthermore, during epoch transition the shards remain ``balanced'' (Theorem $5$~\cite{zamani2018rapidchain}), \ie,  the size of each shard is $O(n/m)$.
\hfill \qed
\end{proof}

\begin{restatable}{thrm}{rcpersistence}\label{rc:persistence}
\rc satisfies persistence in our system model for constant-adaptive adversaries with $f<n/3$ and bounded epoch transitions.
\end{restatable}
\begin{proof}
The consensus protocol in \rc achieves safety if the shard has no more than $t<1/2$ fraction of Byzantine parties (\cite{zamani2018rapidchain}, Theorem $2$). Hence, the statement follows from Lemma \ref{rc:honest-shards}.
\hfill \qed
\end{proof}

\begin{thrm}\label{rc:liveness}
\rc satisfies liveness in our system model for constant-adaptive adversaries with $f<n/3$ and bounded epoch transitions.
\end{thrm}
\begin{proof}
To estimate the liveness of \rc, we need to examine the following subprotocols: (i) {\tt Consensus}, (ii) {\tt CrossShard}, (iii) {\tt DRG}, and (iv) {\tt Divide2Shards}.

The consensus protocol in \rc achieves liveness if the shard has less than $\frac{n}{2m}$ Byzantine parties (Theorem $3$~\cite{zamani2018rapidchain}). Thus, liveness is guaranteed during {\tt Consensus} (Lemma~\ref{rc:honest-shards}).

Furthermore, the final committee is $\frac{1}{2}$-honest with high probability. Hence, the final committee will route each transaction to the corresponding output shard. 
We assume transactions will reach all relevant honest parties via a gossip protocol. \rc employs IDA-gossip protocol, which guarantees message delivery to all honest parties (Lemma $1$ and Lemma $2$~\cite{zamani2018rapidchain}). 
From Assumption~\ref{rc:cross-claim}, the protocol that handles cross-shard transactions satisfies safety even under a Byzantine leader.
Hence, all ``dummy'' transactions will be created and eventually delivered. Since the consensus protocol within each shard satisfies liveness, the ``dummy'' transactions of the input shards will become stable. Consequently, the ``dummy'' transaction of the output shard will become valid and eventually stable (consensus liveness). Thus, {\tt CrossShard} satisfies liveness.
\com{Note that the adversary cannot attack the liveness of the protocol by spamming the system with partially valid transactions since each transaction costs a fee.}

During epoch transition, {\tt DRG} satisfies liveness \cite{feldman1987practical}. Moreover, {\tt Divide2Shards} allows only for a constant number of leave/join/move operations and thus terminates in a constant number of rounds.
\hfill \qed
\end{proof}

\begin{thrm}\label{rc:consistency}
\rc satisfies consistency in our system model for constant-adaptive adversaries with $f<n/3$ and bounded epoch transitions.
\end{thrm}
\begin{proof}
In every epoch, each shard is $\frac{1}{2}$-honest; hence, the adversary cannot double-spend and consistency is satisfied.

Nevertheless, to prove consistency is satisfied across shards, we need to prove that cross-shard transactions are atomic. 
{\tt CrossShard} in \rc ensures that the ``dummy'' transaction of the output shard becomes valid only if all ``dummy'' transactions are stable in the input shards. If a ``dummy'' transaction of an input shard is rejected, the ``dummy'' transaction of the output shard will not be executed, and all the accepted ``dummy'' transactions will just transfer the value of the input UTXOs to other UTXOs that belong to the output shard.
This holds because the protocol satisfies safety even under a Byzantine leader (Assumption~\ref{rc:cross-claim}).

Lastly, the adversary cannot revert the chain of a shard and double-spend an input of the cross-shard transaction after the transaction is accepted in all relevant shards because consistency with each shard and persistence (Theorem \ref{omni:persistence}) hold. Suppose persistence holds with probability $p$. Then, the probability the adversary breaks consistency in a cross-shard transaction is the probability of successfully double-spending in one of the relevant to the transaction shards, hence $1-p^v$ where $v$ is the average size of transactions. Since $v$ is constant, consistency holds with high probability, given persistence holds with high probability.
\hfill \qed
\end{proof}

Similarly to \ol, to calculate the scaling factor of \rc, we need to evaluate the following protocols of the system: (i) {\tt Consensus}, (ii) {\tt CrossShard}, (iii) {\tt DRG}, and (iv) {\tt Divide2Shards}.

\begin{restatable}{lem}{rcshardscale}\label{rc:shard-scale}
The scaling factors of Consensus are $\omega_m=O(\frac{n}{m})$, $\omega_s=O(\frac{1}{m})$, and $\omega_c=O(\frac{n}{m})$.
\end{restatable}
\begin{proof}
From Lemma \ref{rc:shard-size}, the expected number of parties in a shard is $O(n/m)$. The consensus protocol of \rc has quadratic to the number of parties' communication complexity. Hence, the communication factor {\tt Consensus} is $O(\frac{n}{m})$.
The verification complexity (computational factor) collapses to the communication complexity. The space  factor is $O(\frac{1}{m})$, as each party maintains the ledger of the assigned shard for the epoch. 
\end{proof}

\begin{restatable}{lem}{rccrossscale}\label{rc:cross-scale}
The communication and computational factors of CrossShard are both $\omega_m=\omega_c=O(v\frac{n}{m})$, where $v$ is the average size of transactions.
\end{restatable}
\begin{proof}
During the execution of the protocol, the interaction between the input and output shards is limited to the leader, who creates and routes the ``dummy'' transactions.
Hence, the communication complexity of the protocol is dominated by the consensus within the shards. 
For an average size of transactions $v$, the communication factor is $O(vn/m+v)=O(vn/m)$ (Lemma \ref{rc:shard-size}).
Note that this bound holds for the worst case, where transactions have $v-1$ inputs and a single output while all UTXOs belong to different shards.

For each cross-shard transaction, each party of the input and output shards queries the verification oracle once. Hence, the computational factor is $O(vn/m)$.
The protocol does not require any verification across shards, thus the only storage requirement per party is to maintain the ledger of its own shard. \hfill\qed
\end{proof}

\begin{restatable}{lem}{rcepochscale}\label{rc:epoch-scale}
The communication factor of Divide2Shards is $O(\frac{R\cdot n}{m^2})$.
\end{restatable}
\begin{proof}
The number of join/leave and move operations is constant per epoch, denoted by $k$. 
Further, each shard is $\frac{1}{2}$-honest (Lemma \ref{rc:honest-shards}) and has size $O(\frac{n}{m})$ (Lemma \ref{rc:shard-size}); these guarantees hold as long as the number of epochs is $o(n)$. 

Each party changing shards receives the new shard's ledger of size $T/m$ by $O(n/m)$ parties in the new shard. Thus the total communication complexity at this stage is $O(\frac{T}{m}\cdot\frac{n}{m})$, hence the communication factor is $O(\frac{T}{m^2})=O(\frac{R\cdot e}{m^2})$, where $R$ is the number of rounds in each epoch and $e$ the number of epochs since genesis. Since $e=o(n)$, the communication  factor is $O(\frac{R\cdot n}{m^2})$.

% To determine the space overhead factor of the reconfiguration across $l$ epochs, we need to estimate the expected number of different shards each party has been assigned to during these epochs.
% We will show that the expected number of ledgers each party needs to maintain is constant even for $l=n$ epochs.

% The probability a party has not been assigned to a specific shard is $\big(1-\frac{k}{mn} \big)^l \approx 1-e^{-\frac{kl}{nm}} = 1-e^{-\frac{k}{m}}$. The expected number of shards a party has been to after $l$ epochs is $m\big( 1-e^{-\frac{k}{m}} \big) \leq m \frac{k}{m} = k$ which is constant. Thus, the space overhead factor is $k\frac{n}{m}=O(\frac{n}{m})$.
\hfill \qed
\end{proof}

% Note that for $e>n$ epochs, Theorem \ref{rc:epoch-scale} does not hold.

\begin{thrm}\label{rc:scalability}
\rc satisfies scalability in our system model for constant-adaptive adversaries with $f<n/3$ and bounded epoch transitions, with communication and computational factor $O(n/m)$ and space factor $O(1/m)$, where $n=O(m \log m)$, assuming epoch size $R=O(m)$.
\end{thrm}
\begin{proof}
{\tt Consensus} has on expectation communication and computational factors bounded by $O(n/m)$ and space factor $O(1/m)$ (Lemma~\ref{rc:shard-scale}). These bounds are similar in {\tt CrossShard} where the communication and computational factors are bounded by  $O(vn/m)$ (Lemma \ref{rc:cross-scale}), where $v$ is constant (see  Section \ref{sec:limit}).

During epoch transitions, the communication factor dominates: In {\tt DRG} $\omega_m=O(\frac{n}{m})$ (Lemma \ref{rc:random-scale}) while in {\tt Divide2Shards}   $\omega_m=O(\frac{n\cdot R}{m^2})$ (Lemma \ref{rc:epoch-scale}).
Thus for $R=O(m)$, the communication factor during epoch transitions  is $O(n/m)$.

Overall, \rc's expected scaling factors are as follows: $ \omega_m=\omega_c= O(n/m)=O(\log m)$ and $\omega_s=O(1/m)$, where the  equation holds for $n=c'm \log m$ (Lemma \ref{rc:honest-shards}).
\end{proof}

\begin{restatable}{thrm}{rcthroughput}\label{rc:throughput}
In \rc, the throughput factor is $\sigma=\mu \cdot \tau \cdot \dfrac{m}{v} < \frac{\mu \cdot \tau \cdot f(n)}{v}$ with $f(n) = \frac{n}{c'\log(\frac{n}{c'\log(n)})}$ with $c' = \frac{c}{p}$ and constant $c$ from corollary~\ref{cor:shard-size}.
%\mu \cdot \tau \cdot \dfrac{n}{\log n} \cdot \dfrac{(a-p)^2}{2+a-p} \cdot \dfrac{1}{v}$.\todo{to be changed - Antoine}
\end{restatable}
\begin{proof}
At most $v$ shards are affected per transaction -- when each transaction has $v-1$ inputs and one output, and all belong to different shards. Therefore, $m'<m/v$.
From Lemma \ref{lem:shard-size} and corollary~\ref{cor:upper-bound}, 
$ m < f(n)$.
Therefore, $\sigma < \frac{\mu \cdot \tau \cdot f(n)}{v}$.
%$m<\frac{n}{\log n} \cdot \frac{(a-p)^2}{2+a-p}$. Therefore, $\sigma < \mu \cdot \tau \cdot \frac{n}{\log n} \cdot \frac{(a-p)^2}{2+a-p} \cdot \frac{1}{v}$.
\hfill \qed
\end{proof}

In \rc, the consensus protocol is synchronous and thus not practical. We estimate the throughput factor irrespective of the chosen consensus, to provide a fair comparison to other protocols. 
We notice that both \rc and \ol have the same throughout factor when $v$ is constant. 

We provide an example of the throughput factor in case the employed consensus is the one suggested in \rc. In this case, we have $a=1/2$, $p=1/4$ (hence $p/a=2/3$), $\mu<1/2$ (Theorem 1~\cite{zamani2018rapidchain}), and $\tau=1/8$ ($4$ rounds are needed to reach consensus for an honest leader, and the leader will be honest every two rounds on expectation~\cite{abraham2018hotstuff}.). Note that $\tau$ can be improved by allowing the next leader to propose a block even if the previous block is not yet accepted by all honest parties; however, we do not consider this improvement. Because of the values of $p$ and $a$ we can compute $c\simeq2.6$, thus $c'\simeq 10.4$.
Hence, for $v=5$, we have throughput factor:
  \[\sigma < \frac{1}{2} \cdot \frac{1}{8} \cdot \frac{1}{5} \cdot \frac{1}{10.4} \frac{n}{\log (\frac{n}{10.4\log n})}  =\frac{n}{832 \log (\frac{n}{10.4\log n})}  \]
 %\[\sigma < \frac{1}{2} \cdot \frac{1}{8} \cdot \frac{n}{\log n} \cdot \frac{(\frac{1}{2}-\frac{1}{3})^2}{2+\frac{1}{2}-\frac{1}{3}} \cdot \frac{1}{5} =\frac{n}{6240 \log n}  \]

%------------------------CHAINSPACE------------------------
 \subsection{\cs}
 \cs is a sharding protocol introduced by Al-Bassam et al.~\cite{al2018chainspace} that operates in the permissioned setting. 
The main innovation of \cs is on the application layer. Specifically, \cs presents a sharded, UTXO-based distributed ledger that supports smart contracts. Furthermore, limited privacy is enabled by offloading computation to the clients, who need to only publicly provide zero-knowledge proofs that their computation is correct.
\cs focuses on specific aspects of sharding; epoch transition or reconfiguration of the protocol is not addressed.
Nevertheless, the cross-shard communication protocol, namely S-BAC, is of interest as a building block to secure sharding. 

\paragraph{S-BAC protocol.}\label{app:sbac}
S-BAC is a shard-led cross-shard atomic commit protocol used in \cs.
In S-BAC, the client submits a transaction to the input shards.
Each shard internally runs a BFT protocol to
tentatively decide whether to accept or abort the transaction locally and broadcasts its local decision to other  shards that take part in the transaction. 
If the transaction fails locally (\eg, is a double-spend), then the shard generates pre-abort(T), whereas if the transaction succeeds locally the shard generates pre-accept(T) and changes the state of the input to `locked'.
After a shard decides to pre-commit(T), it waits to collect responses from other participating shards,
and commits the transaction if all shards respond with pre-accept(T), or aborts the transaction if at least one shard announces pre-abort(T).
Once the shards decide, they send their decision (accept(T) or abort(T)) to the client and the output shards. 
If the decision is accept(T), the output shards generate new `active' objects and the input shards change the input objects to `inactive'. 
If an input shard's decision is abort(T), all input shards unlock the input objects by changing their state to `active'.

S-BAC, just like Atomix, is susceptible to replay attacks~\cite{sonnino2019replay}. To address this problem, sequence numbers are added to the transactions, and output shards generate dummy objects during the first phase (pre-commit, pre-abort). More details and security proofs can be found on~\cite{sonnino2019replay}, as well as a hybrid of Atomix and S-BAC called Byzcuit.

% \subsubsection{Discussion.} 
%Chainspace does not address the issue of identity creation and dividing participants into shards. Therefore, it is not suitable for permissionless settings as consistency is not guaranteed by design for shards that do not maintain an honest supermajority ($3f+1$ parties where at most $f$ are Byzantine) and Chainspace does not explain how such shards can be created. Furthermore, Chainspace argues for external to the system punishment for identities that provably break consistency (in compromised shards) which is only applicable in a permissioned setting. For these reasons, Chainspace is out of scope for this work, as the system is insufficient by design to maintain a distributed secure sharded transaction ledger in a permissionless setting.